\documentclass{article}%
\usepackage{amsfonts}
\usepackage{amsmath}
\usepackage{amssymb}
\usepackage{todonotes}
\usepackage{graphicx}%
\usepackage{lipsum}  
\usepackage{float} 
\usepackage{braket} 
\usepackage{charter} 

\newtheorem{theorem}{Theorem}

\newtheorem{definition}[theorem]{Definition}

\newtheorem{proposition}[theorem]{Proposition}
\newtheorem{remark}[theorem]{Remark}

\newenvironment{proof}[1][Proof]{\noindent\textbf{#1.} }{\ \rule{0.5em}{0.5em}}

\newcommand{\wt}[1]{\widetilde{#1}}
\newcommand{\Xn}[1]{X^{\left( #1 \right)} }
\newcommand{\Xtn}[1]{\wt X^{\left( #1 \right)} }
\newcommand{\X}[2]{X^{\left( #1 \right)} (#2)}
\newcommand{\Xt}[2]{\wt{X}^{\left( #1 \right)} (#2)}

\newcommand{\angl}[1]{\left\langle #1 \right\rangle}

\begin{document}
\begin{center}
{\Large{\bf Supersymmetric generalized power functions}}\\
\bigskip 
Mathieu Ouellet\footnote{\texttt{mathieu.ouellet@uqtr.ca}} and S\'{e}bastien Tremblay\footnote{\texttt{sebastien.tremblay@uqtr.ca}}\\
\bigskip 
D\'epartement de math\'ematiques et d'informatique, \\ Universit\'e du Qu\'ebec, Trois-Rivi\`{e}res, QC, G9A 5H7, Canada\\
\bigskip \bigskip

{\bf Abstract}
\end{center}
Complex-valued functions defined on a finite interval $[a,b]$ generalizing power functions of the type $(x-x_0)^n$ for $n\geq 0$ are studied. These functions called $\Phi$-generalized powers, $\Phi$ being a given nonzero complex-valued function on the interval, were considered to contruct a general solution representation of the Sturm-Liouville equation in terms of the spectral parameter \cite{kravchenko2008, kravporter2010}. The $\Phi$-generalized powers can be considered as a natural basis functions for the one-dimensional supersymmetric quantum mechanics systems taking $\Phi=\psi_0^2$, where the function $\psi_0(x)$ is the ground state wave function of one of the supersymmetric scalar Hamiltonians. Several properties are obtained such as $\Phi$-symmetric conjugate and antisymmetry of the $\Phi$-generalized powers, a supersymmetric binomial identity for these functions, a supersymmetric Pythagorean elliptic (hyperbolic) identity involving four $\Phi$-trigonometric ($\Phi$-hyperbolic) functions as well as a supersymmetric Taylor series expressed in terms of the $\Phi$-derivatives. We show that the first $n$ $\Phi$-generalized powers are a fundamental set of solutions associated with a nonconstant coefficients homogeneous linear ordinary differential equations of order $n+1$. Finally, we present a general solution representation of the stationary Schr\"odinger equation in terms of geometric series where the Volterra compositions of the first type is considered.



\newpage

\section{Introduction}
Spectral parameter power series (SPPS) method was first introduced in 2008 \cite{kravchenko2008} and subsequently generalized in 2010 from a different approach \cite{kravporter2010}. This method has been widely used in the last ten years. One of the reasons is that SPPS has implications from both a mathematical and a physical point of view. Indeed, as illustrated in the following theorem, SPPS consists of a representation for the solutions of the Sturm-Liouville equation as a spectral parameter. The range of applications is therefore considerable.

\begin{theorem}
\label{kravch}
{\bf \cite{kravporter2010}}
Assume that on a finite real interval $[a,b]$, equation $(pu_0')'+qu_0 =0 $ possesses a particular solution $u_0$ such that $u_0^2r$ and $1/(u_0^2 p)$ are continuous
on $[a, b]$. Then the general solution of  $(pu')'+qu=\lambda ru$ on $(a, b)$ has the form 
\begin{equation}
\label{gensol1}
u = c_1u_1 + c_2u_2
\end{equation}
where $c_1$  and $c_2$ are arbitrary complex constants,
\begin{equation}
\label{gensol2}
u_1 =u_0 \sum_{k=0}^{\infty}\frac{\lambda^k }{(2k)!} \Xtn{2k}, \quad \text{and} \quad u_2 = u_0\sum_{k=0}^{\infty} \frac{\lambda^k}{(2k+1)!} \Xn{2k+1}
\end{equation}
with $\Xtn{n}$ and $\Xn{n}$  being defined by the recursive relations

\begin{equation*}
\label{X0}
\Xtn{0} \equiv 1 , \qquad  \Xn{0}\equiv 1,
\end{equation*}
\begin{equation}
\label{XGtdef}
\begin{array}{rcl}
\X{n}{x_0,x}&=&\left\{
                \begin{array}{ll}
                  \displaystyle n\int_{x_0}^x \X{n-1}{x_0,\xi}\frac{1}{u_0^2(\xi)p(\xi)}d\xi,&  n\text{ odd}, \\[2ex]
                   \displaystyle n\int_{x_0}^x \X{n-1}{x_0,\xi}u_0^2(\xi)r(\xi)d\xi,&  n\text{ even},
                \end{array}
              \right.\\ \\

\Xt{n}{x_0,x}&=&\left\{
                \begin{array}{ll}
                  \displaystyle n\int_{x_0}^x \Xt{n-1}{x_0,\xi}u_0^2(\xi)r(\xi)d\xi,& n\text{ odd}, \\*[2ex]
                   \displaystyle n\int_{x_0}^x \Xt{n-1}{x_0,\xi}\frac{1}{u_0^2(\xi)p(\xi)}d\xi,&    n\text{ even},
                \end{array}
              \right.
\end{array}
\end{equation}
where $x_0$ is an arbitrary point in $[a,b]$ such that $p$ is continuous at $x_0$ and $p(x_0)\neq 0$. The two series in (\ref{gensol2}) converge uniformly on $[a,b]$.
\end{theorem}

From the mathematical point of view, SPPS was considered for the study of the transmutation operators (see \cite{begehr, carroll, colton, levitan, marchenko}) in \cite{kravchenko2018, kravchenko2016, kravchenko2013, kravchenko2015}, for the Sturm-Liouville or the Hill equations \cite{blancarte, erbe, khmelnytskaya2014, khmelnytskaya2010, porter2016} as well as for the study of perturbed Bessel equations \cite{castillo2013, castillo2018}.

From the physical point of view, SPPS method was used for acoustic \cite{rabinovich2015}, the heat transfer problem \cite{khmelnytskaya2013}, the one-dimensional quantum scattering \cite{rabinovich2007} as well as optical fibers \cite{castillo2015}. However, SPPS was first derived using pseudoanalytic function theory in \cite{kravchenko2008} and it has has been shown that this theory is closely related to the two-dimensional supersymmetric quantum mechanics (SUSYQM)~\cite{bilodeautremblay}. 

In recent years, SPPS has also proved to be very efficient numerically, see for instance \cite{blancarte, castillo2015, castillo2013, erbe, han2016, khmelnytskaya2012, khmelnytskaya2014, khmelnytskaya2013, kravchenko2008, kravchenko2018, kravchenko2016, kravporter2010, kravchenko2015, castillo2018}. Nevertheless, in this paper we study the generalized power functions $X^{(n)},\wt X^{(n)}$ appearing in the SPPS method from the analytical point of view. Indeed, despite the extensive use of these functions in the  mathematical and  the physical litterature in recent years and their fundamental role in the SPPS method, these generalized power functions have not been studied much analytically; this is the main objective of this paper.

In Section 2, we present an overview of the one-dimensional SUSYQM. The purpose of this section is mainly to illustrate the fundamental role played by the ground state wave function $\psi_0(x)$ in these systems. Indeed, we show that despite the function $1/\psi_0(x)$ cannot be  considered as an eigenfunction, the symmetry $\psi_0(x)\rightarrow 1/\psi_0(x)$ acts as a permutation of the bosonic and the fermionic parts of the SUSYQM system. In Section~3, considering this symmetry the $\Phi$-generalized power functions are formally introduced in agreement with the SPPS method (Theorem \ref{kravch}). The principle of alternation $\Phi(x)$ and $1/\Phi(x)$ in the definition of the $\Phi$-generalized power functions $X^{(n)},\wt X^{(n)}$ appears according to the parity of $n$, where $\Phi(x)=\psi_0^2(x)$ for the SUSYQM system. Many properties of the $\Phi$-generalized power functions are given: 1) the symmetries of the variables permutation of these functions which, among other things, will prevent many numerical calculations for the SPPS method, 2) a binomial identity and 3) a generalized trigonometric Pythagorean identity. Generalized Taylor series are developed in Section 4 for the real-valued functions $\Phi(x)$. Section 5 is devoted to the relations between the Volterra composition of the first type and the generalized power functions. These relationships allow us to construct a general solution representation of the stationary Schr\"odinger equation in terms of the Volterra composition on a finite interval.

\section{One-dimensional supersymmetric quantum \\ mechanics}
\subsection{General results and overview}
Let us consider the simplest one-dimensional SUSYQM system characterised by the existence of charge operators $\mathcal Q,\mathcal Q^\dagger$ that obey the algebra \cite{witten}
\begin{equation}
\label{susyalgebra}
\{\mathcal Q, \mathcal Q^\dagger\}= \mathcal Q \mathcal Q^\dagger+ \mathcal Q^\dagger \mathcal Q=\mathcal H,\qquad \mathcal Q^2=0,\qquad \mathcal Q^{\dagger 2}=0.
\end{equation}
From these equations it is easy to see that $[\mathcal Q,\mathcal H]=0,\ [\mathcal Q^\dagger,\mathcal H]=0$.
The algebra defined in \eqref{susyalgebra} can be achieved by considering the matrices $\mathcal Q=\footnotesize{\left(
\begin{array}{cc}
0 & 0 \\
\mathcal A & 0
\end{array}
\right)}$ and $\mathcal Q^\dagger=\footnotesize{\left(
\begin{array}{cc}
0 & \mathcal A^\dagger \\
0 & 0
\end{array}
\right)}$, where $\mathcal A$ is a linear operator and $\mathcal A^\dagger$ is the adjoint. Combining these matrices with algebra \eqref{susyalgebra} leads to the supersymmetric Hamiltonian $\mathcal H=\footnotesize{\left(
\begin{array}{cc}
\mathcal H_1 & 0 \\
0 & \mathcal H_2
\end{array}
\right)}$, where $\mathcal H_1= \mathcal A^\dagger \mathcal A$ and $\mathcal H_2= \mathcal A \mathcal A^\dagger$. These scalar Hamiltonians $\mathcal H_1, \mathcal H_2$ are both positive semi-definite operators with eigenvalues greater than or equal to zero.

Let us choose the ground state energy of $\mathcal H_1$ to be zero. 
Considering the nodeless ground state wave function $\psi_0^{(1)}(x)$ belonging to the Hamiltonian $\mathcal H_1$, the Schr\"odinger equation
$$
\mathcal H_1\psi_0^{(1)}(x)=-\frac{d^2}{dx^2}\psi_0^{(1)}(x)+V_1(x)\psi_0^{(1)}(x)=0
$$
gives us the possibility to obtain the potential $V_1$ in terms of the ground state wave function $\psi_0^{(1)}$:
$$
V_1(x)=\frac{[\psi_0^{(1)}(x)]^{\prime\prime}}{\psi_0^{(1)}(x)}.
$$
It is now simple to factorize the Hamiltonian $H_1$ as $\mathcal A^\dagger \mathcal A$ for $\mathcal A=\frac{d}{dx}+W(x)$ and $\mathcal A^\dagger=-\frac{d}{dx}+W(x)$, where the function $W(x)$ is generally refered to as the {\em superpotential}. Hence, by using the operators $\mathcal A$ and $\mathcal A^\dagger$ it is possible to write $V_1(x)$ in terms of the superpotential:
\begin{equation}
\label{V1}
V_1(x)=W^2(x)-W^\prime(x).
\end{equation}
One solution for $W(x)$ in terms of the ground state wave function of $\mathcal H_1$ is
\begin{equation}
\label{superpot}
W(x)=-\frac{[\psi_0^{(1)}(x)]^\prime}{\psi_0^{(1)}(x)}.
\end{equation}
This solution is obtained by recognizing that once we satisfy $\mathcal A\psi_0^{(1)}=0$, i.e. 
$$
\psi_0^{(1)}(x)=N\exp\left(-\int_{-\infty} ^ x W(\xi)d\xi\right),\quad N\text{ a normalization constant,}
$$
we automatically have a solution to $\mathcal H_1\psi_0^{(1)}= \mathcal A^\dagger \mathcal A\psi_0^{(1)}=0$.

Let us now consider the other scalar Hamiltonian $\mathcal H_2= \mathcal A \mathcal A^\dagger$ obtained by reversing the order of $\mathcal A$ and $\mathcal A^\dagger$ in $\mathcal H_1$. A little simplification shows that the operator $\mathcal H_2$ is in fact a Hamiltonian corresponding to a new potential $V_2(x)$:
\begin{equation}
\label{V2}
\mathcal H_2=-\frac{d^2}{dx^2}+V_2(x),\qquad V_2(x)=W^2(x)+W^\prime (x).
\end{equation}
The energy eigenvalues and the wave functions of the Hamiltonians $\mathcal H_1$ and $\mathcal H_2$ are related in the following way. 
For $n\geq 0$ we have
$$
E_0^{(1)}=0,\qquad E_n^{(2)}=E_{n+1}^{(1)},
$$
$$
\psi_n^{(2)}=\big[E_{n+1}^{(1)} \big]^{-\frac{1}{2}} \mathcal A\psi_{n+1}^{(1)}, \quad \text{and}\quad \psi_{n+1}^{(1)}=\big[E_{n}^{(2)} \big]^{-\frac{1}{2}} \mathcal A^\dagger \psi_{n}^{(2)}.
$$

\subsection{The bosonic-fermionic permutation symmetry from the ground state wave function}

Let us consider a symmetry of the supersymmetric Hamiltonian. We observe that from equations \eqref{V1} and \eqref{V2} the Hamiltonains $\mathcal H_1$ and $\mathcal H_2$ can be expressed in terms of the superpotential $W(x)$ and its first derivative. In other words, from equation \eqref{superpot}  the partner Hamiltonians $\mathcal H_1$ and $\mathcal H_2$ depend on the ground state wave function $\psi_0^{(1)}(x)$ of $\mathcal H_1$. Introducing now the operator $\mathcal R_{\psi_0^{(1)}}$ which transforms $\psi_0^{(1)}(x)\rightarrow 1/\psi_0^{(1)}(x)$ we find
$$
\mathcal R_{\psi_0^{(1)}}[\psi_0^{(1)}]=\wt \psi_0^{(1)}=\frac{1}{\psi_0^{(1)}},\quad \mathcal R_{\psi_0^{(1)}}[W]=\wt W=-W,
$$
and
$$
\mathcal R_{\psi_0^{(1)}}[\mathcal H_1]=\wt {\mathcal H}_1=\mathcal H_2,\quad \mathcal R_{\psi_0^{(1)}}[\mathcal H_2]=\wt {\mathcal H}_2=\mathcal H_1,
$$
where the tilde notation will be used in what follows to represent the functions or the operators on which the operator $\mathcal R_{\psi_0^{(1)}}$ is applied. Hence, the operator $\mathcal R_{\psi_0^{(1)}}$ is just a permutation of the two scalar Hamiltonians  $\mathcal H_1$ and $\mathcal H_2$ in the supersymmetric Hamiltonian $\mathcal H$.

In particular, we note that since $\mathcal A\psi_0^{(1)}(x)=0$, we have
$$
\mathcal R_{\psi_0^{(1)}}[\mathcal A\psi_0^{(1)}(x)]=\wt {\mathcal A}\,\wt \psi_0^{(1)}(x)=-\mathcal A^\dagger \left(1/\psi_0^{(1)}(x)\right)=0,
$$
i.e. $\mathcal H_2\big(1/\psi_0^{(1)}(x)\big)=0$. However, if we suppose that $\psi_0^{(1)}(x)$ is normalized, then the function $1/\psi_0^{(1)}(x)$ cannot be normalized. Therefore, the function $1/\psi_0^{(1)}(x)$ is not considered as the ground state wave function of $\mathcal H_2$; this implies that $E_0^{(1)}$ is the only nondegenerate state of the supersymmetric Hamiltonian $\mathcal H$. 

Despite that  $1/\psi_0^{(1)}(x)$ is not a ground state wave function of $\mathcal H_2$, the symmetry induced by $\mathcal R_{\psi_0^{(1)}}$ remains for the partner Hamiltonians for all the eigenvalues $E_n^{(2)}=E_{n+1}^{(1)}$ of the system for $n\geq 0$. Moreover, the SUSYQM system is completely determined by the ground state wave function of $\mathcal H_1$.

In what follows, we will study complex-valued functions generalizing powers functions of the type $(x-x_0)^n$ for $n\geq 0$, taking into account the permutation symmetry $\mathcal R_{\psi_0^{(1)}}$ of the one-dimensional SUSYQM system.

\section{Definition and properties of the $\Phi$-generalized power functions}
In the frame work of this paper, in order to simplify the problem, we shall study the functions $X^{(n)}$ and $\widetilde X^{(n)}$ appearing in the representation for the general solution of $u''+qu=\lambda u$, i.e. for $p=r=1$ in Theorem  \ref{kravch}. This leads us to the following definition.

\begin{definition}
\label{defX}
Let $\Phi (x)$ be a given continuous nonzero complex-valued function on $[a,b]$ and $x_0\in [a,b]$. The $\Phi $-power functions (or the generalized power functions) are defined iteratively by $\Xt{0}{x_0,x}\equiv 1$, $\X{0}{x_0,x}\equiv 1$ and
\begin{equation}
\label{Xtdef}
\begin{array}{l}
\X{n}{x_0,x}=n\displaystyle\int_{x_0}^x \X{n-1}{x_0,\xi}\Big(\Phi (\xi)\Big)^{(-1)^n}d\xi, \\*[3ex]
\Xt{n}{x_0,x}=n\displaystyle\int_{x_0}^x \Xt{n-1}{x_0,\xi}\left(\frac{1}{\Phi (\xi)}\right)^{(-1)^n}d\xi,
\end{array}
\end{equation}
where $n>0$. We say that the $\Phi $-power function $\widetilde X^{(n)}$ is the $\Phi $-conjugate of $X^{(n)}$ (and conversely).
\end{definition}
\begin{remark}
We observe that the $\Phi $-conjugaison of the generalized power functions are obtained from the transformation $\mathcal R_{\Phi}[X^{(n)}] = \wt X^{(n)}$.
\end{remark}

\begin{remark}
 In the case $\Phi \equiv 1$ we have $\Xt{n}{x_0,x}=\X{n}{x_0,x}$ and these functions are equal to $(x-x_0)^n$. In the general case, the $\Phi$-power functions  $X^{(n)}$, $\widetilde X^{(n)}$ are distinct and not necessarily polynomial functions.
\end{remark}

\label{remsch}
When considering the general solution representation of the one-dimensional stationary Schr\"odinger equation $\mathcal H_1\psi(x)=-\frac{d^2}{dx^2}\psi(x)+V_1(x)\psi(x)=\lambda\psi(x)$
from the point of view of Theorem \ref{kravch}, i.e. $p=-1$, $r=1$ and $q=V_1=[\psi_0^{(1)}]^{\prime\prime}/\psi_0^{(1)}$, then the generalized power functions $X_{1}^{(n)}$, $\wt X_{1}^{(n)}$ associated with this Schr\"odinger equation in Theorem \ref{kravch} are related to the generalized power functions $X^{(n)}$, $\wt X^{(n)}$ in \eqref{Xtdef} by taking $\Phi=[\psi_0^{(1)}]^2$. Since $p=-1$ we have $X_1^{(n)}=(-1)^n X^{(n)}$ and $\wt X_1^{(n)}=(-1)^{n+1} \wt X^{(n)}$. Hence, from \eqref{gensol1} and \eqref{gensol2} the general solution representation of $\mathcal H_1\psi(x)=\lambda \psi(x)$ can be written in terms of the functions $X^{(n)}$, $\wt X^{(n)}$ (up to a minus sign absorbed by the arbitrary complex constants $c_1,c_2$) defined in \eqref{Xtdef} by
$$
\psi=c_1 \psi_0 \sum_{k=0}^{\infty}\frac{\lambda^k }{(2k)!} \Xtn{2k}+c_2 \psi_0\sum_{k=0}^{\infty} \frac{\lambda^k}{(2k+1)!} \Xn{2k+1}.
$$

Applying now operator $\mathcal R_{\Phi}$ on equation $\mathcal H_1\psi(x)=\lambda \psi(x)$, we obtain the general solution representation $\wt \psi(x)$ of the supersymmetric partner Schr\"odinger equation $\mathcal H_2\wt \psi(x)=\lambda \wt \psi(x)$, i.e.
$$
\wt \psi=\frac{\wt c_1}{\psi_0} \sum_{k=0}^{\infty}\frac{\lambda^k }{(2k)!} \Xn{2k}+\frac{\wt c_2}{\psi_0}\sum_{k=0}^{\infty} \frac{\lambda^k}{(2k+1)!} \Xtn{2k+1}.
$$


\begin{theorem}
\label{symmX}
Let $X^{(n)}$ and $\widetilde X^{(n)}$ be $\Phi $-generalized powers. These functions are $\Phi $-symmetric conjugate for even $n$ and antisymmetric for odd $n$, i.e.
\begin{equation}
\label{conjRel1}
\Xt{n}{x_0,x} = \X{n}{x,x_0}, \qquad n \text{ even}
\end{equation}
and
\begin{equation}
\label{conjRel2}
\begin{array}{ll}
\Xt{n}{x_0,x}= -\Xt{n}{x,x_0}, & n \text{ odd},\\[2ex]
\X{n}{x_0,x}= -\X{n}{x,x_0}, & n \text{ odd}.
\end{array}
\end{equation}
\end{theorem}

\begin{proof}
Let $f_1(x),\ldots , f_n(x)$ be continuous complex-valued functions in $[a,b]$ such that $f_i(x) \neq 0$ for all $x\in [a,b]$. We consider Chen's iterated path integrals introduced in \cite{chen1977} and defined by
\begin{equation*}
\label{ittDef}
\int_{x_0}^x f_1(\xi)d\xi \cdots f_n(\xi) d\xi := \int_{x_0}^x\left( \int_{x_0}^\xi f_1(\xi)d\xi \cdots f_{n-1}(\xi) d\xi \right) f_n(\xi) d\xi \qquad \text{for }n>1.
\end{equation*}
When $n=0$, set the integral to be $1$. The following property was obtained by Chen for iterated integrals \cite{chen1977} :
\begin{equation}
\int_{x}^{x_0} f_1(\xi)d\xi \cdots f_n(\xi) d\xi = (-1)^n \int_{x_0}^x f_n(\xi)d\xi \cdots f_1(\xi) d\xi.
\label{chen}
\end{equation}
Now let us consider the particular case \begin{equation}
\label{deffn}
f_n=\left(\frac{1}{\Phi }\right)^{(-1)^n} \qquad \text{for }n\geq 1. 
\end{equation}
Then equation \eqref{chen} becomes
$$
\frac{1}{n!}\Xt{n}{x,x_0}=
\left\{
\begin{array}{ll}
-\displaystyle\frac{1}{n!}\Xt{n}{x_0,x},& n\text{ odd}\\*[2ex]
\displaystyle\frac{1}{n!}\X{n}{x_0,x},& n\text{ even},
\end{array}
\right.
$$
which demonstrate property \eqref{conjRel1} and property \eqref{conjRel2} for the $\Phi $-conjugate $\widetilde X^{(n)}$. Property \eqref{conjRel2} for the $X^{(n)}$ is obtained by considering the case $f_n(x)=\Phi ^{(-1)^n}$ instead of functions \eqref{deffn}.
\end{proof}

\begin{remark}
Theorem \ref{symmX} is particularly interesting when SPPS method is used for numerical analysis. Indeed, when a SPPS numerical calculations is performed both set of functions $\Xt{n}{x_0,x}$, $\X{n}{x_0,x}$ are calculated for some $N>0$ and $n=1,\ldots, N$. For even numbers $n$, this theorem gives us the possibility to obtain all the functions $\Xt{n}{x_0,x}$ from the functions $\X{n}{x,x_0}$. Consequently, numerical calculations of the SPPS could be improved by about $25\,\%$ by using this result.
\end{remark}

\begin{remark}
A corollary is obtained from Theorem \ref{symmX}: for $n\geq 1$ the $\Phi$-conjugate power functions are related by the formulas
\begin{align*} 
\X{n}{x_0,x}=n\int_{x_0}^x \frac{\Xt{n-1}{\xi,x}}{\Phi(\xi)} d\xi, & \quad n \text{ even}\\
\Xt{n}{x_0,x}=n\int_{x_0}^x \Phi(\xi) \X{n-1}{\xi,x} d\xi, &\quad n \text{ odd}.
\end{align*}
Indeed, considering first the even $n$, from the definition \ref{defX} we have
$$
\wt X^{(n)}(x,x_0)=n\int_{x}^{x_0} \frac{\wt X^{(n-1)}(x,\xi)}{\Phi(\xi)}d\xi=n\int_{x_0}^{x} \frac{\wt X^{(n-1)}(\xi,x)}{\Phi(\xi)}d\xi,
$$
where in the last equality is obtained by interchanging the order of integration and using Theorem \ref{symmX}. However, since $n$ is even $\wt X^{(n)}(x,x_0)=X^{(n)}(x_0,x)$ again from Theorem \ref{symmX}. Hence, we obtain the the desired result when $n$ is even. The second formula when $n$ is odd is easily proven in a similar way. 
\end{remark}

\begin{theorem}[Binomial identity for even $n$]
\label{binomial}
Let $\Phi (x)$ be a given continuous nonzero complex-valued function on $[a,b]$ and $x_0\in [a,b]$. Then for a finite and an even $n\geq 2$ the $\Phi $-powers functions \eqref{Xtdef} satisfy the following identity:
\begin{equation*}
\label{binomX}
\sum_{k=0}^n (-1)^k \binom{n}{k}X^{(k)}\,\wt X^{(n-k)}=0
\end{equation*}
or equivalently
$$
\wt X^{(n)}=\sum_{k=1}^n (-1)^{k+1} \binom{n}{k}X^{(k)}\,\wt X^{(n-k)}.
$$
\end{theorem}
\begin{proof}
By integrating by parts for an even $n\geq 2$ we have
\begin{align*}
X^{(n)}(x_0,x)&=n\int_{x_0}^x X^{(n-1)}(x_0,\xi)\Phi (\xi) d\xi \\*[2ex]
&=n\Big[X^{(n-1)}\wt X^{(1)}-(n-1)\int_{x_0}^x \wt X^{(1)}X^{(n-2)}\frac{1}{\Phi } d\xi\Big],
\end{align*}
where for simplicity the arguments of the functions have been omitted in the second line (and for the rest of the proof). In particular, we observe that when $n=2$ we obtain $X^{(2)}=2X^{(1)}\wt X^{(1)}-\wt X^{(2)}$ which prove the identity in this case.

Now by integrating by parts again for $n\geq 3$, we find
$$
\int_{x_0}^x \wt X^{(1)}X^{(n-2)}\frac{1}{\Phi } d\xi=\frac{1}{2}\wt X^{(2)}X^{(n-2)}-\frac{1}{2}(n-2)\int_{x_0}^x \wt X^{(2)}X^{(n-3)}\Phi  d\xi
$$
such that
$$
X^{(n)}=n\wt X^{(1)}X^{(n-1)}-\frac{1}{2}n(n-1)\wt X^{(2)}X^{(n-2)}+\frac{1}{2}n(n-1)(n-2)\int_{x_0}^x \wt X^{(2)}X^{(n-3)}\Phi  d\xi.
$$
Pursuing these calculations, i.e. applying $n-1$ integration by parts for an even n, we finally obtain
$$
X^{(n)}=\sum_{k=1}^{n-1}(-1)^{k+1} \frac{n!}{k!(n-k)!} \wt X^{(k)}X^{(n-k)}-n\int_{x_0}^x \wt X^{(n-1)}X^{(0)}\frac{1}{\Phi } d\xi,
$$
where by definition the last term is just $\wt X^{(n)}$ and this complete the proof.
\end{proof}
\begin{remark}
Considering the same proof when $n$ is odd and finite, we obtain the trivial identity 
$$
\sum_{k=0}^n (-1)^k \binom{n}{k}X^{(k)}X^{(n-k)}=0.
$$
\end{remark}
\begin{remark}
From the usual binomial expansion, Theorem \ref{binomial} can be seen {\em from a purely symbolic point of view} as 
$$
(X-\wt X)^{(n)}=0 \quad \text{for even }n.
$$ 
\end{remark}

A corollary can be obtained from Theorem \ref{binomial} to obtain a generalization of the elliptic and the hyperbolic Pythagorean identities. For that, let us define the $\Phi $-generalized sine and cosine functions as

\begin{align*}
& C(x_0,x)=\sum_{j=0}^{\infty} \frac{(-1)^j }{(2j)!}\X{2j}{x_0,x}, &\wt{C}(x_0,x)=\sum_{k=0}^{\infty} \frac{(-1)^k }{(2k)!}\Xt{2k}{x_0,x},\\
&S(x_0,x)=\sum_{j=0}^{\infty} \frac{(-1)^j}{(2j+1)!} \X{2j+1}{x_0,x},
&\wt{S}(x_0,x)=\sum_{k=0}^{\infty} \frac{(-1)^k}{(2k+1)!} \Xt{2k+1}{x_0,x},
\end{align*}
and the $\Phi $-generalized hyperbolic functions  as
\begin{align*}
&Ch(x_0,x)=\sum_{j=0}^{\infty} \frac{\X{2j}{x_0,x}}{(2j)!}, &\wt{Ch}(x_0,x)=\sum_{k=0}^{\infty} \frac{\Xt{2k}{x_0,x}}{(2k)!},\\
&Sh(x_0,x)=\sum_{j=0}^{\infty} \frac{\X{2j+1}{x_0,x}}{(2j+1)!} ,
&\wt{Sh}(x_0,x)=\sum_{k=0}^{\infty} \frac{\Xt{2k+1}{x_0,x}}{(2k+1)!} .
\end{align*}
These series are convergent. Indeed, from definition \ref{defX} we find
$$
\frac{X^{(2j)}(x_0,x)}{(2j)!}=\int_{x_0}^x\Phi (\xi_1)\int_{x_0}^{\xi_1}\frac{1}{\Phi (\xi_2)}\int_{x_0}^{\xi_2}\Phi (\xi_3)\int_{x_0}^{\xi_3}\cdots \int_{x_0}^{\xi_{2j-1}}\frac{1}{\Phi (\xi_{2j})}d\xi_{2j}\cdots d\xi_1
$$
and
$$
\frac{X^{(2j+1)}(x_0,x)}{(2j+1)!}=\int_{x_0}^x\frac{1}{\Phi (\xi_1)}\int_{x_0}^{\xi_1}\Phi (\xi_2)\int_{x_0}^{\xi_2}\frac{1}{\Phi (\xi_3)}\int_{x_0}^{\xi_3}\cdots \int_{x_0}^{\xi_{2j}}\frac{1}{\Phi (\xi_{2j+1})}d\xi_{2j+1}\cdots d\xi_1
$$
such that on the interval $[a,b]$ we have
$$
\left|X^{(2j)}(x_0,x)\right|\leq \Big(\max |\Phi (x)|\Big)^j \Big(\max \left|\frac{1}{\Phi (x)}\right|\Big)^j\big|b-a\big|^{2j}=c^j
$$
and
$$
\left|X^{(2j+1)}(x_0,x)\right|\leq \Big(\max |\Phi (x)|\Big)^{j} \Big(\max \left|\frac{1}{\Phi (x)}\right|\Big)^{j+1}\big|b-a\big|^{2j+1}=|b-a| \Big(\max \left|\frac{1}{\Phi (x)}\right|\Big) c^j
$$
for $c=\Big(\max |\Phi (x)|\Big) \Big(\max \left|\frac{1}{\Phi (x)}\right|\Big)\big|b-a\big|^{2}$. Hence, the series $\sum_{j=0}^\infty \frac{c^j}{(2j)!}<\infty$ and $|b-a|\big(\max \left|\frac{1}{\Phi (x)}\right|\big)\sum_{j=0}^\infty \frac{c^j}{(2j+1)!}<\infty$ such that the $\Phi$-trigonometric functions and the $\Phi$-hyperbolic functions converge absolutely and uniformly from the Weierstrass M-test.

\begin{theorem}[Supersymmetric Pythagorean identities]
\label{trigo}
The $\Phi $-trigonometric functions satisfy the following generalized Pythagorean trigonometric identity
\begin{equation*}
C(x_0,x)\,\wt C(x_0,x)+S(x_0,x)\,\wt{S}(x_0,x)=1.
\end{equation*}
Moreover, the $\Phi$-hyperbolic functions satisfy the generalized hyperbolic identity
\begin{equation*}
Ch(x_0,x)\,\wt{Ch}(x_0,x)-Sh(x_0,x)\,\wt{Sh}(x_0,x)=1.
\end{equation*}
\end{theorem}
\begin{proof}
Let us consider first the Pythagorean trigonometric identity. From the infinite double sum property
\begin{equation}
\label{dinf}
\sum_{j=0}^\infty\sum_{k=0}^\infty c_{k,j}=\sum_{p=0}^\infty \sum_{q=0}^p c_{q,p-q},
\end{equation}
we obtain
\begin{align*}
C\wt{C}&=\sum_{j=0}^{\infty} \sum_{k=0}^{\infty} \frac{(-1)^{j+k} }{(2j)!(2k)!}X^{(2j)}\wt X^{(2k)}=\sum_{p=0}^\infty\sum_{q=0}^{p}\frac{(-1)^p}{(2q)!(2p-2q)!}X^{(2p-2q)}\wt X^{(2q)} \\*[2ex]
&=1+\sum_{p=0}^\infty\sum_{q=0}^{p+1}\frac{(-1)^{p+1}}{(2q)!(2p+2-2q)!}X^{(2p+2-2q)}\wt X^{(2q)} \\*[2ex]
&= 1+\sum_{p=0}^\infty \Big(\sum_{q=0}^{p}\frac{(-1)^{p+1}}{(2q)!(2p+2-2q)!}X^{(2p+2-2q)}\wt X^{(2q)}+ \frac{(-1)^{p+1} }{(2p+2)!}\wt X^{(2p+2)}\Big).
\end{align*}
Using again \eqref{dinf} for $S\wt{S}$ we find
\begin{align*}
S\wt{S}=\sum_{p=0}^\infty\sum_{q=0}^{p}\frac{(-1)^p}{(2q+1)!(2p+1-2q)!}X^{(2p+1-2q)}\wt X^{(2q+1)} 
\end{align*}
such that
\begin{align*}
C\wt{C}+S\wt{S}&= 1+\sum_{p=0}^\infty \Big(\sum_{q=0}^{p}\frac{(-1)^{p+1}}{(2q)!(2p+2-2q)!}X^{(2p+2-2q)}\wt X^{(2q)}+\\*[2ex]&\qquad +\frac{(-1)^p}{(2q+1)!(2p+1-2q)!}X^{(2p+1-2q)}\wt X^{(2q+1)} + \frac{(-1)^{p+1} }{(2p+2)!}\wt X^{(2p+2)}\Big)\\*[2ex]
&= 1+\sum_{p=0}^\infty  \frac{(-1)^{p+1}}{(2p+2)!}\Big(\sum_{q=0}^{p}\frac{(2p+2)!}{(2q)!(2p+2-2q)!}X^{(2p+2-2q)}\wt X^{(2q)}+\\*[2ex]&\qquad-\frac{(2p+2)!}{(2q+1)!\big((2p+2)-(2q+1)\big)!}X^{(2p+1-2q)}\wt X^{(2q+1)} + \wt X^{(2p+2)}\Big)\\*[2ex]
&=1+\sum_{p=0}^\infty  \frac{(-1)^{p+1}}{(2p+2)!}\sum_{q=0}^{2p+2}(-1)^{q}\binom{2p+2}{q} X^{(2p+2-q)}\wt X^{(q)}.
\end{align*}
However, since $2p+2$ is even and greater or equal than $2$ the last summation over $q$ is always zero from Theorem \ref{binomial} and this complete the proof for the first identity.

The hyperbolic identity can be proved in a similar way.
\end{proof}

The $\Phi$-trigonometric functions are illustrated in Figure \ref{f1a} for the real-valued function $\Phi(x)=(1+x)^2$. The generalized Pythagorean trigonometric identity of Theorem \ref{trigo} are then shown in Figure \ref{f1b} for the same function $\Phi(x)$. This case can be seen as a perturbation of the standard trigonometric functions, i.e. $\Phi\equiv 1$ around $x=0$. In the same way, we illustrate the $\Phi$-trigonometric functions in Figures \ref{f2a} and the generalized Pythagorean trigonometric identity in Figures \ref{f2b} for the real-valued function $\Phi(x)=\sqrt{\cosh x}$. In both case, we set $x_0=0$. We limit ourselves to these two examples, however several other functions $\Phi(x)$ have been considered.

\begin{figure}[H]
\centering
\caption{\label{f1a}Graphics of the $\Phi$-trigonometric functions for $\Phi=(1+x)^2$ and $x_0=0$}
\includegraphics[width=1\textwidth]{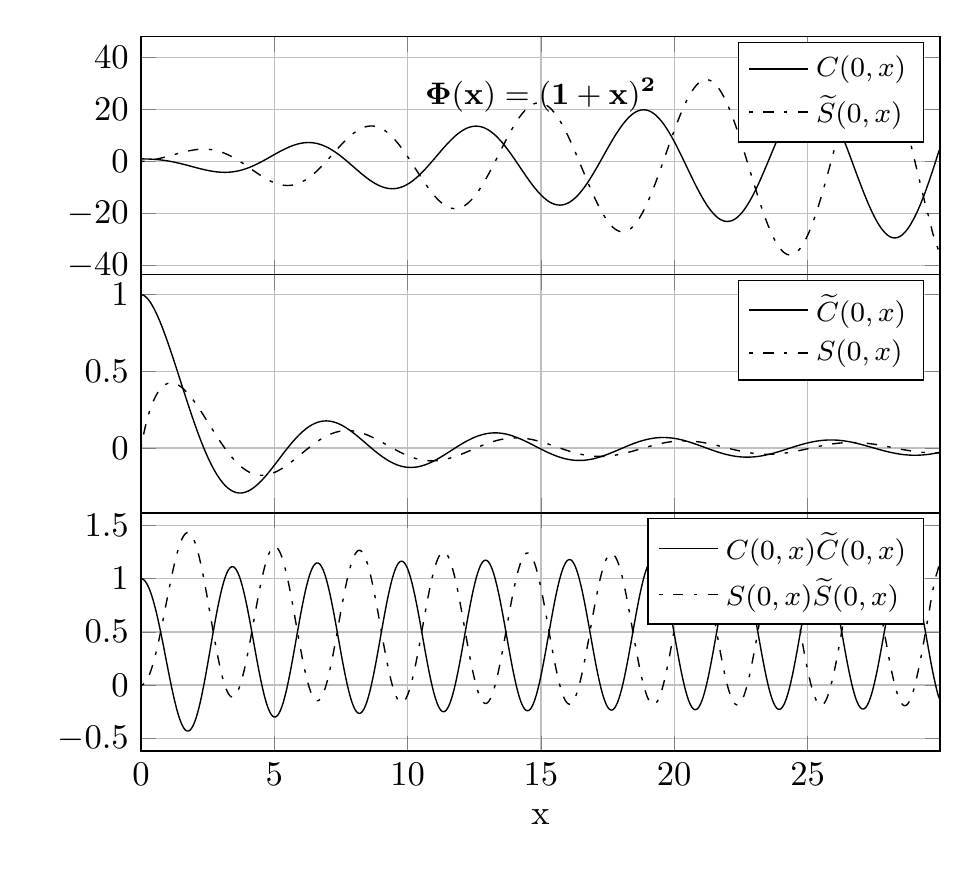}
 \end{figure}     
 
 \begin{figure}[H]
\centering
\caption{\label{f1b}Generalized Pythagorean trigonometric identity for $\Phi=(1+x)^2$ and $x_0=0$ in the phase space of the $\Phi$-trigonometric functions}
\includegraphics[width=1\textwidth]{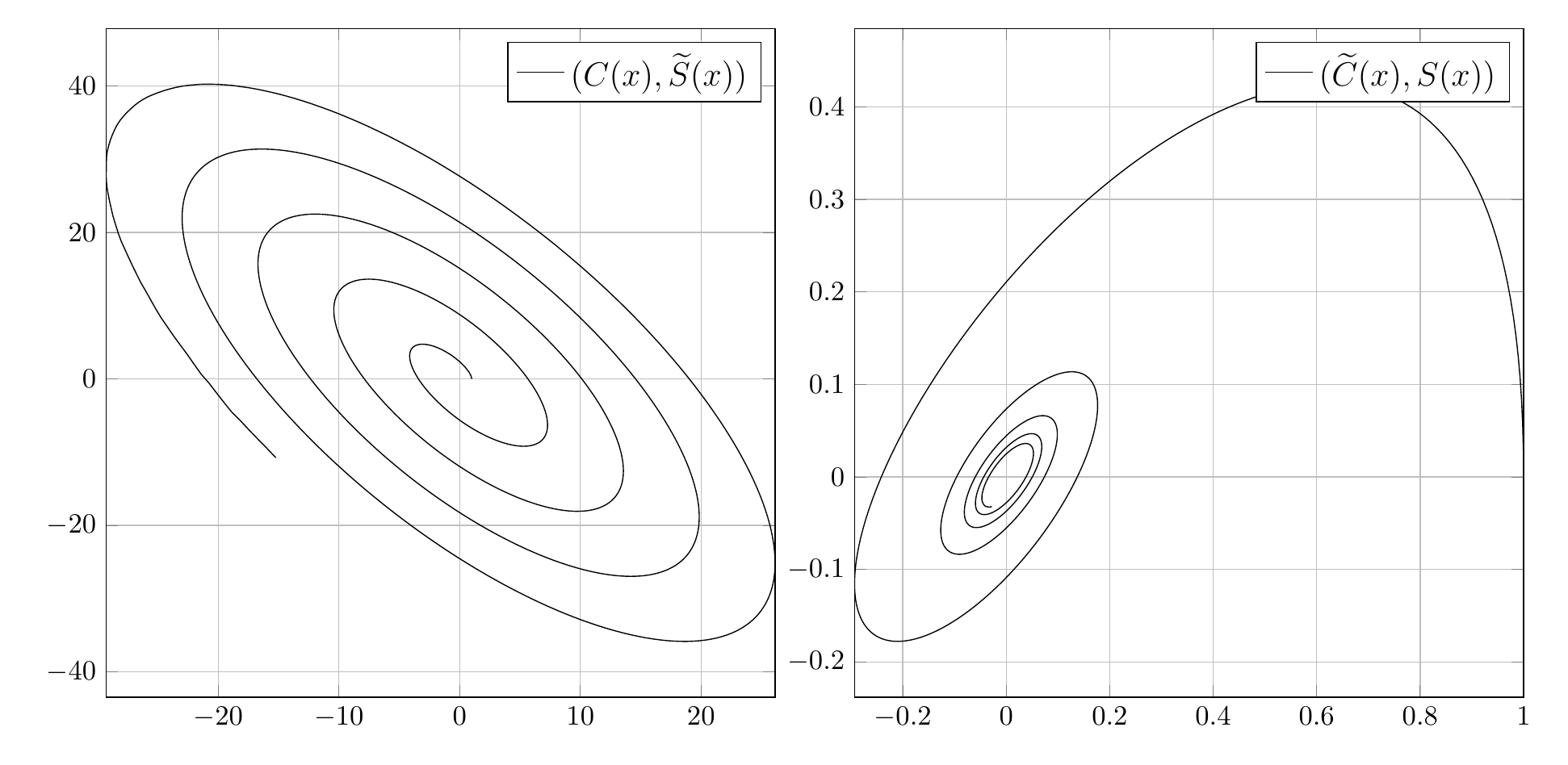}
 \end{figure}     
 
\begin{figure}[H]
\centering
\caption{\label{f2a}Graphics of the $\Phi$-trigonometric functions for $\Phi=\sqrt{\cosh x}$ and $x_0=0$}
\includegraphics[width=1\textwidth]{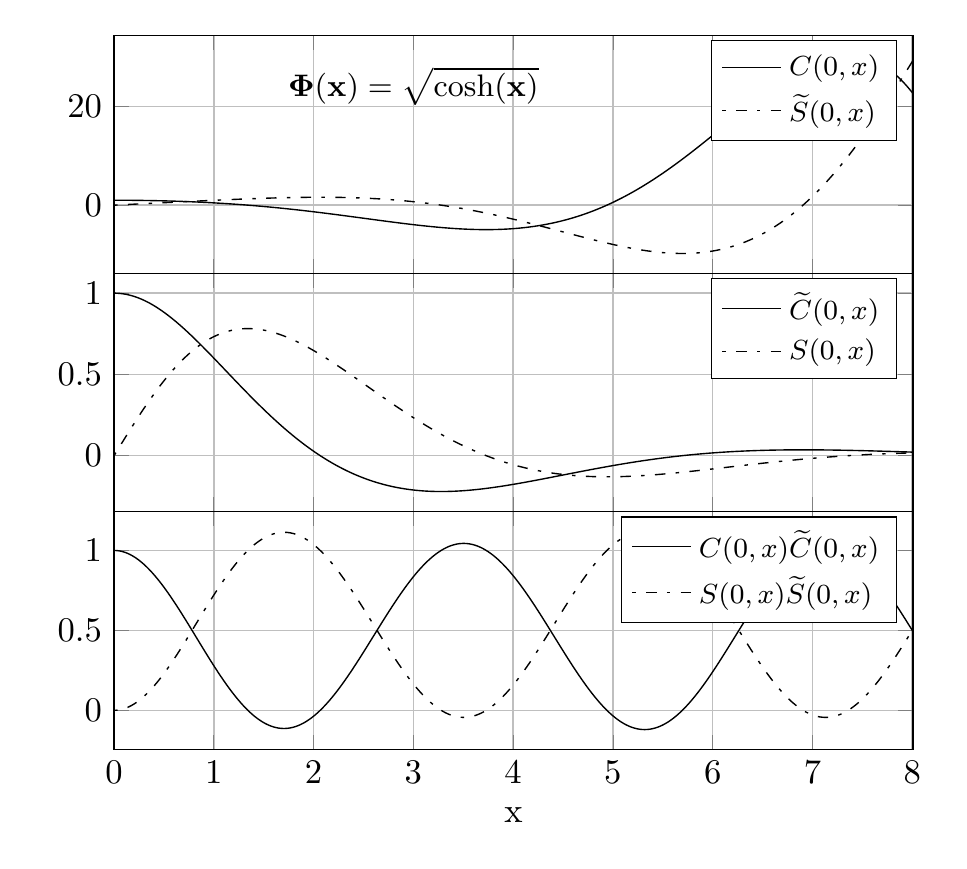}
 \end{figure}     
 
 \begin{figure}[H]
\centering
\caption{\label{f2b}Generalized Pythagorean trigonometric identity for $\Phi=\sqrt{\cosh x}$ and $x_0=0$ in the phase space of the $\Phi$-trigonometric functions}
\includegraphics[width=1\textwidth]{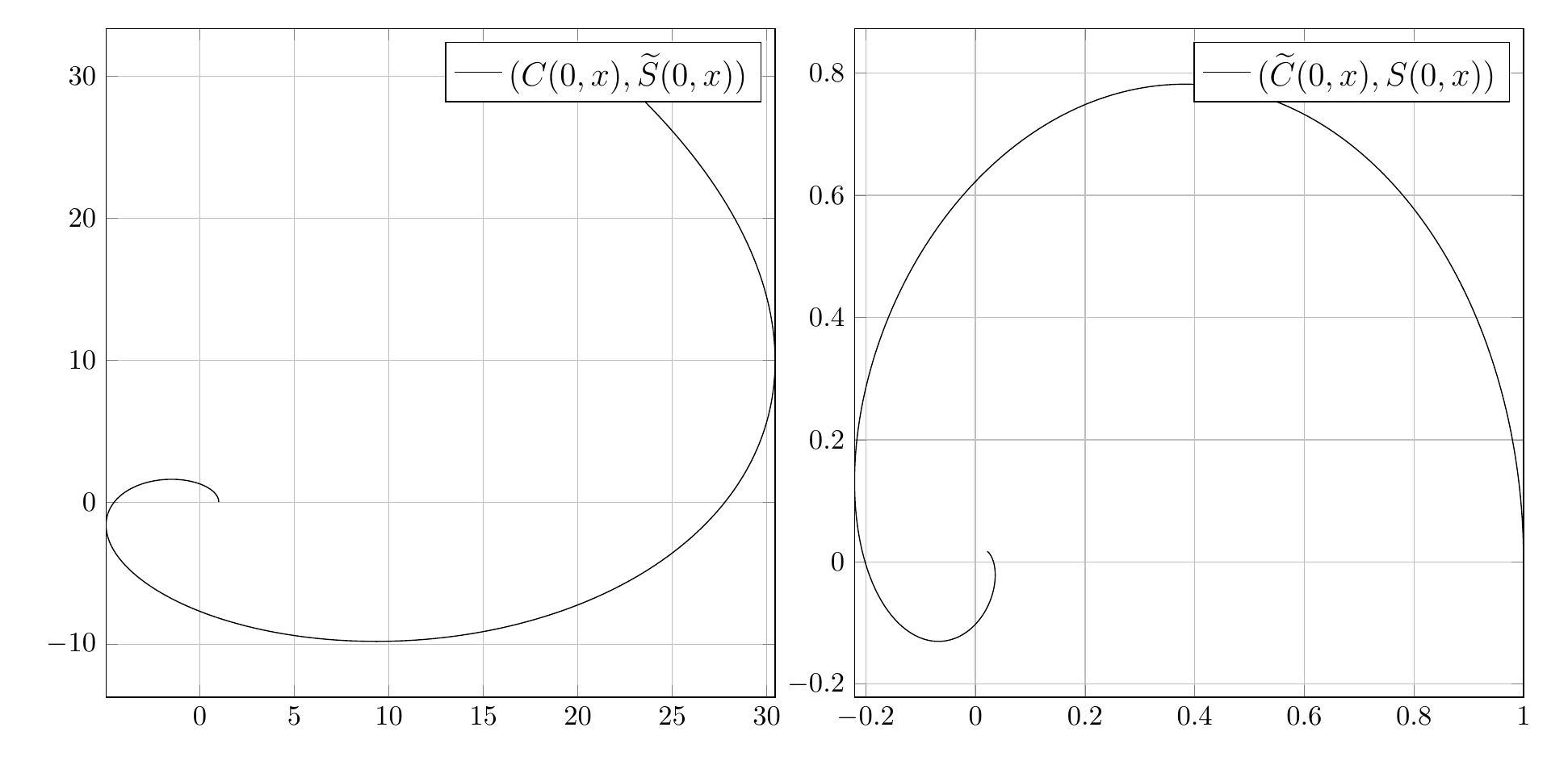}
 \end{figure}

In what follows the set of $\Phi$-power functions $\{\Xn{2k}\}_{k=0}^{\infty} \cup \{\Xtn{2k+1}\}_{k=0}^{\infty}$ (as well as the $\Phi$-conjugate set $\{\Xtn{2k}\}_{k=0}^{\infty} \cup \{\Xn{2k+1}\}_{k=0}^{\infty}$) will be of particular interest. For simplicity, we define the functions $\mathcal Y_n(x_0,x)$, $\wt {\mathcal Y}_n(x_0,x)$ as 

\begin{equation}
\label{defY}
\mathcal Y_{n}(x_0,x):=\left\{
                \begin{array}{ll}
                  \Xt{n}{x_0,x}, &  n\text{ odd}, \\*[2ex]
                   \X{n}{x_0,x},& n\text{ even},
                \end{array}
              \right. 
\end{equation}
and
\begin{equation}
\label{defYt}
\wt{\mathcal Y}_{n}(x_0,x):=\left\{
                \begin{array}{ll}
                  \X{n}{x_0,x}, &  n\text{ odd}, \\*[2ex]
                   \Xt{n}{x_0,x},& n\text{ even},
                \end{array}
              \right.
\end{equation}
for $n\geq 0$.

It was shown that generalized powers posses the property of completeness in $L^2(a,b)$ \cite{kravchenko2011completeness, kravtrembmorelos}. Indeed, for $x_0$ an arbitrary fixed point in $[a,b]$ it has been shown that when the point $x_0$ coincides with one of the extreme of the interval, then the infinite system of functions $\{\Xn{2k}\}_{k=0}^{\infty}$ is complete in $L^2(a,b)$. Under the same condition, the completeness of $\{\Xtn{2k+1}\}_{k=0}^{\infty}$ is guaranteed in $L^2(a,b)$. Meanwile, when $x_0$ is an interior point of the interval the system of functions $\{\mathcal Y_{k}\}_{k=0}^\infty$ is complete in $L^2(a,b)$. For instance, in the case $\Phi \equiv 1$ and choosing $x_0=0$ on the interval $[0,1]$, we find
\begin{equation}
\label{xpairs}
X^{(0)}(x)=1,\quad X^{(2)}(x)=x^2,\quad X^{(4)}(x)=x^4,\ldots
\end{equation}
and
\begin{equation}
\label{ximpairs}
\Xt{1}{x}=x,\quad \Xt{3}{x} =x^3,\quad \Xt{5}{x} =x^5,\ldots
\end{equation}
Both systems of monomials are complete in $L^2(0,1)$ due to the M\"untz theorem \cite[p.275]{davis}. Hence, the system of fonctions $\{\Xn{2n}\}_{n=0}^{\infty}$ and $\{\Xtn{2n+1}\}_{n=0}^{\infty}$ represent a direct generalization of the system of monomials \eqref{xpairs}, \eqref{ximpairs} if instead of $\Phi \equiv 1$ an arbitrary sufficiently smooth and nonvanishing function $\Phi $ is choosen.

Let us now introduce a $\Phi$-derivative which will be used in what follows:

\begin{equation*}
Dh(x)=\Phi (x)\frac{d}{dx}h(x) \qquad \text{and}\qquad \wt Dh(x)=\frac{1}{\Phi (x)}\frac{d}{dx}h(x),
\end{equation*}
where $h(x)$ is an arbitrary complex-valued function well defined on the interval of interest. For $D^{(0)}h(x)\equiv h(x)$ and $\wt D^{(0)}h(x)\equiv h(x)$
 and the higher-order $\Phi $-derivatives are defined alternately as
\begin{equation*}
D^{(k)}h(x) =D\left(\wt D^{(k-1)}h(x) \right)\qquad \text{and} \qquad \wt{D}^{(k)}h(x) =\wt{D}\left( D^{(k-1)}h(x) \right).
\end{equation*}

These definitions of generalized derivatives agree with many aspect with the standard derivative applied on  functions of the type $(x-x_0)^n$. For instance, we have that $D^{(k)}\X{k}{x_0,x}=\wt{D}^{(k)}\Xt{k}{x_0,x}=k!$. More generally, we have the following proposition.

\begin{proposition}
\label{DX}
When $n,k\in \mathbb Z_{\geq 0}$ are of the same parity and $n\geq k$, we have
\begin{align}
D^{(k)} \Xn{n}(x_0,x)&= \displaystyle \frac{n!}{(n-k)!}\Xn{n-k}(x_0,x),\label{DkXn} \\*[2ex]
\wt D^{(k)} \Xtn{n}(x_0,x)&=  \displaystyle \frac{n!}{(n-k)!}\Xtn{n-k}(x_0,x). \label{DtkXtn}
\end{align}
On the other, when $n,k\in \mathbb Z_{\geq 0}$ are not of the same parity and $n> k$, we obtain
\begin{align}
\wt D^{(k)} \Xn{n}(x_0,x)&=  \displaystyle \frac{n!}{(n-k)!}\Xn{n-k}(x_0,x),\label{DtkXn}\\[2ex] D^{(k)} \Xtn{n}(x_0,x)&=  \displaystyle \frac{n!}{(n-k)!}\Xtn{n-k}(x_0,x). \label{DkXtn}
\end{align}
\end{proposition}
\begin{proof}
We begin by considering induction on $n$ for equation \eqref{DkXn}. The cases $n=0,1$ are trivially verified. Let us suppose that $n$ is even and equation \eqref{DkXn} is valid for $0\leq k\leq n$ and $k$ even. Then for $1\leq k\leq n+1$ and $k$ odd we obtain

\begin{align}
\label{proofDX}
D^{(k)}X^{(n+1)}&=D\big(\wt D^{(k-1)} X^{(n+1)}\big)=D^{(k-1)}\big(D X^{(n+1)}\big)\\*[2ex]&=D^{(k-1)}\left((n+1)\Phi X^{(n)}\frac{1}{\Phi }\right)=(n+1)\frac{n!}{(n+1-k)!}X^{(n-k+1)}, \nonumber
\end{align}
where the second equality in \eqref{proofDX} is valid for $k$ odd. This prove the first equation \eqref{DkXn} when $n$ is even. The case $n$ odd (and $k$ odd) of equation \eqref{DkXn} is prove in similar way as well as its $\Phi$-conjugate counterpart \eqref{DtkXtn}.

Equations \eqref{DtkXn} and \eqref{DkXtn} are proven in a similar way.
\end{proof}

\begin{remark}
Despite these results, other combinations of $\Phi $-derivatives applied on the $\Phi $-power functions do not give such simple results. For instance, we have $D^{(2)}X^{(3)}=6X^{(1)}-6\left(\frac{1}{\Phi }\right)\left(\frac{\Phi '}{\Phi }\right)X^{(2)}$.
\end{remark}

\begin{remark}
Considering the $\Phi$-derivatives of the generalized trigonometric and the hyperbolic functions, we easily show that
$DS(x_0,x) = C(x_0,x)$, $\wt DC(x_0,x)= -S(x_0,x)$, $D[Sh(x_0,x)] = Ch(x_0,x)$ and $\wt D[Ch(x_0,x)]= Sh(x_0,x)$.
\end{remark}

\section{Supersymmetric Taylor series from the $\Phi $-power functions}

In this section, and only in this section, we suppose that the function $\Phi(x)$ is a real-valued function on $[a,b]$. We will treat the approximation of functions by generalized Taylor series using $\Phi $-generalized powers. Taylor series have already been considered from the iterated integrals $X^{(n)},\wt X^{(n)}$ for the Sturm-Liouville equation, see \cite{kravtrembmorelos}; here we consider a more general setting and a different approach.

In what follows we shall have frequent occasions to use Wronskian so that it will be convenient to introduce some notations. The functions $f_1(x),\ldots, f_n(x)$ being of class $C^n$, we set
$$
\mathcal W[f_1(x),f_2(x),\ldots,f_n(x)]\equiv \left|
\begin{array}{cccc}
f_1(x) & f_2(x) & \cdots & f_n(x)\\
f_1'(x) & f_2'(x) & \cdots & f_n'(x)\\
\vdots & \vdots & \ddots & \vdots \\
f_1^{(n)}(x) & f_2^{(n)}(x) & \cdots & f_n^{(n)}(x)
\end{array}
\right|.
$$
Moreover, since in what follows we will pay special attention to the sets \\ $\{\mathcal Y_0(x_0,x),\mathcal Y_1(x_0,x),\ldots, \mathcal Y_n(x_0,x)\}$ and $\{\wt{\mathcal Y}_0(x_0,x),\wt{\mathcal Y}_1(x_0,x),\ldots, \wt{\mathcal Y}_n(x_0,x)\}$, we set as a particuliar case
$$
\mathcal W_n(x)\equiv \mathcal W\big[\mathcal Y_0(x_0,x),\mathcal Y_1(x_0,x),\ldots, \mathcal Y_n(x_0,x)\big]
$$
and
$$
\wt{\mathcal W}_n(x)\equiv \mathcal W\big[\wt{\mathcal Y}_0(x_0,x),\wt{\mathcal Y}_1(x_0,x),\ldots, \wt{\mathcal Y}_n(x_0,x)\big].
$$

\begin{proposition} 
\label{linind}
Let $x_0$ be a given point in the interval $(a,b)$ of $\Phi $ and the two sets of fonctions $\mathcal S_n=\{\mathcal Y_0(x_0,x),\ldots, \mathcal Y_n(x_0,x)\}$, $\wt{\mathcal{S}}_n=\{\wt 
{\mathcal Y}_0(x_0,x),\ldots, \wt{\mathcal Y}_n(x_0,x)\}$ are of class $C^n(a,b)$. Then the sets of functions $\mathcal S_n$ and $\wt{\mathcal S}_n$ are linearly independant.


\end{proposition}
\begin{proof}
Let us first consider the Wronskian matrix at the point $x=x_0$, i.e. calculate  the matrix $\mathcal Y_j^{(i)}(x_0,x_0)$ for $0\leq i,j\leq n$. We have that $\mathcal Y_j(x_0,x_0)=\delta_{i,0}$, where $\delta_{i,k}$ is the Kronecker delta. The first line and column of the Wronskian matrix being known, we now consider the strictly upper diagonal terms $\mathcal Y_k^{(j)}(x_0,x)$, where $1\leq j<k\leq n$. Considering the derivatives, these terms can be expanded as a summation of the functions $\mathcal Y_r(x_0,x)$, $\wt{\mathcal Y}_r(x_0,x)$ for $1\leq r\leq k-1$. Hence, evaluated at $x=x_0$ we find $\mathcal Y_k^{(j)}(x_0,x)=0$ and the Wronskian matrix is lower triangular. By similar calculations we find that Wronskian matrix $\wt{\mathcal Y}_j^{(i)}(x_0,x_0)$ of the functions $\wt{\mathcal S}_n$ at $x=x_0$ is lower triangular.

On the diagonal of the Wronskian matrices we obtain
\begin{equation}
\label{diagw1}
\mathcal Y_m^{(m)}(x_0,x_0)= \left\{
                \begin{array}{ll}
                  \displaystyle  m!\,\Phi (x_0), &  m\text{ odd}, \\*[2ex]
                   \displaystyle m!, & m\text{ even},
                \end{array}
              \right. 
\end{equation}
and
\begin{equation}
\label{diagw2}
\wt{\mathcal Y}_m^{(m)}(x_0,x_0) = \left\{
                \begin{array}{ll}
                  \displaystyle  m!\,\Phi ^{-1}(x_0), &  m\text{ odd}, \\*[2ex]
                   \displaystyle m!, & m\text{ even},
                \end{array}
              \right. 
\end{equation}
for $0\leq m\leq n$. This can be shown by induction on $n$. Indeed, equations \eqref{diagw1}, \eqref{diagw2} are satisfed for $n=0,1$. For an even $n>0$ suppose that equations \eqref{diagw1}, \eqref{diagw2} are valid for $0\leq m\leq n-1$. We have
$$
\mathcal Y_n^{(n)}=\frac{d^{n-1}}{dx^{n-1}}\mathcal Y_n'=\frac{d^{n-1}}{dx^{n-1}}\Big(n\Phi \wt{\mathcal Y}_{n-1}\Big)=n\sum_{k=0}^{n-1}\binom{n-1}{k} \Phi ^{(n-1-k)}\wt{\mathcal Y}_{n-1}^{(k)}.
$$
As has been shown below, the terms $\wt{\mathcal Y}_{n-1}^{(k)}$ evaluated at $x=x_0$ are zero for $k=0,\ldots, n-2$ such that
$$
\mathcal Y_n^{(n)}(x_0,x_0)=n\Phi (x_0)\wt{\mathcal Y}_{n-1}^{(n-1)}(x_0,x_0)=n\Phi (x_0)(n-1)!\Phi ^{-1}(x_0)=n!.
$$
The case $n$ odd and greater than $1$ can be shown in a similar way to complete the proof of \eqref{diagw1}. Equation \eqref{diagw2} is also proved by induction in a similar way.

Hence, the Wronskians at $x=x_0$ is given by
$$
\mathcal W_n(x_0)=\prod_{k=0}^n \mathcal Y_k^{(k)}(x_0,x_0)= \left\{\begin{array}{ll}
\alpha_n\Big(\sqrt{\Phi (x_0)}\Big)^{n+1}, & n \text{ odd},\\*[2ex]
\alpha_n \Big(\sqrt{\Phi (x_0)}\Big)^{n}, & n \text{ even},
\end{array}
\right.
$$
and
$$
\wt{\mathcal W}_n(x_0)=\left\{\begin{array}{ll}
\alpha_n\Big(\displaystyle\frac{1}{\sqrt{\Phi (x_0)}}\Big)^{n+1}, & n \text{ odd},\\*[3ex]
\alpha_n \Big(\displaystyle\frac{1}{\sqrt{\Phi (x_0)}}\Big)^{n}, & n \text{ even},
\end{array}
\right.
$$
where
$$
\alpha_n:=\prod_{k=0}^n k!.
$$

Now since the set of $n+1$ functions of $\mathcal S_n$ ($\wt{\mathcal S}_n$) are $n$ times differentiable over the interval $(a,b)$ with $W_n(x_0)\neq 0$ ($\wt{W}_n(x_0)\neq 0$) for some $x_0\in (a,b)$, the functions of $\mathcal S_n$ ($\wt{\mathcal S}_n$) are linearly independent from the Abel's identity. 
\end{proof}

\begin{proposition}
\label{fundset}
Let $x_0$ be a given point in the interval $[a,b]$ of $\Phi $. Then the fundamental set of solutions for the ordinary differential equation $D^{(n+1)}y(x)=0$ of order $n+1$ is
\begin{equation}
\label{fundset1}
\left\{
\begin{array}{ll}
\mathcal S_{n}=\{\mathcal Y_0(x_0,x),\ldots, \mathcal Y_{n}(x_0,x)\}, & n\text{ odd}, \\*[2ex]
\wt{\mathcal S}_{n} =\{\wt {\mathcal Y}_0(x_0,x),\ldots, \wt{\mathcal Y}_{n}(x_0,x)\}, & n\text{ even}.
\end{array}
\right.
\end{equation}
Moreover, the fundamental set of solutions for the ordinary differential equation $\wt D^{(n+1)}y(x)=0$ of order $n+1$ is
\begin{equation}
\label{fundset2}
\left\{
\begin{array}{ll}
\wt{\mathcal S}_{n}=\{\wt{\mathcal Y}_0(x_0,x),\ldots, \wt{\mathcal Y}_{n}(x_0,x)\}, & n\text{ odd},\\*[2ex]
\mathcal S_{n}=\{\mathcal Y_0(x_0,x),\ldots, \mathcal Y_{n}(x_0,x)\}, & n\text{ even}.
\end{array}
\right.
\end{equation}
\end{proposition}
\begin{proof}
Let us first consider the case $n$ odd. For  $0\leq k\leq n$ we  obtain
$$
D^{(n+1)}\mathcal Y_k=\left\{
\begin{array}{ll}
D^{(n+1-k)}\big(\wt D^{(k)}\wt X^{(k)}\big)=D^{(n+1-k)}k!=0, & k\text{ odd},  \\*[2ex]
D^{(n+1-k)}\big(D^{(k)}X^{(k)}\big)=D^{(n+1-k)}k!=0, & k\text{ even} ,
\end{array}
\right.
$$
where proposition \ref{DX} was used. Similarly, we show that $\wt D^{(n+1)}\wt{\mathcal Y}_k=0$ when $n$ is odd and $0\leq k\leq n$. 

For $n$ even and $0\leq k\leq n$ we have
$$
D^{(n+1)}\wt{\mathcal Y}_k=\left\{
\begin{array}{ll}
D^{(n+1-k)}\big(D^{(k)}X^{(k)}\big)=D^{(n+1-k)}k!=0, & k\text{ odd},  \\*[2ex]
D^{(n+1-k)}\big(\wt D^{(k)}\wt X^{(k)}\big)=D^{(n+1-k)}k!=0, & k\text{ even} .
\end{array}
\right.
$$
In the same way, we find that $\wt D^{(n+1)}\mathcal Y_k=0$ when $n$ is even and $0\leq k\leq n$.
\end{proof}

To demonstrate Proposition \ref{linind} it was sufficient to show that $\mathcal W_n(x_0)\neq 0$ and $\wt{\mathcal W}_n(x_0)\neq 0$ .  Nervertheless, in the following it will be useful to obtain an explicit form of the Wronskians $\mathcal W_n(x)$ and $\wt {\mathcal W}_n(x)$ for all $x$ in the interval $[a,b]$. For that we use explicitly the Abel's identity. From Proposition \ref{fundset} the set $\mathcal S_n=\{\mathcal Y_0(x_0,x),\ldots ,\mathcal Y_n(x_0,x)\}$ is a fundamental set of solutions for the ordinary differential equations $D^{(n+1)}y(x)=0$ with $n$ odd and $\wt D^{(n+1)}y(x)=0$ with $n$ even. By developping these equations in terms of the usual derivatives, i.e.
$$
D^{(n+1)}y(x)=a_{n+1}(x)y^{(n+1)}(x)+a_n(x)y^{(n)}(x)+\cdots +a_{1}(x)y'(x)=0
$$
and
$$
\wt D^{(n+1)}y(x)=\wt a_{n+1}(x)y^{(n+1)}(x)+\wt a_n(x)y^{(n)}(x)+\cdots +\wt a_{1}(x)y'(x)=0
$$
we find
$$
\begin{array}{ll}
a_{n+1}(x)=1,\qquad a_n(x)=-\displaystyle\frac{n+1}{2}\frac{\Phi '(x)}{\Phi (x)}, & n\text{ odd}, \\*[3ex]
\wt a_{n+1}(x)=\displaystyle\frac{1}{\Phi (x)},\qquad \wt a_n(x)=-\displaystyle\frac{n}{2}\frac{\Phi '(x)}{\Phi (x)^2}, & n\text{ even},
\end{array}
$$
such that $\mathcal W_n(x)=\mathcal W_n(x_0)\exp\Big(-\displaystyle\int_{x_0}^x\frac{a_n(\xi)}{a_{n+1}(\xi)}d\xi\Big)$ by the Abel's identity. Explicitly, we find
\begin{equation}
\label{Wn}
\mathcal W_n(x)=\left\{\begin{array}{ll}
\alpha_n\Big(\sqrt{\Phi (x)}\Big)^{n+1}, & n \text{ odd},\\*[2ex]
\alpha_n \Big(\sqrt{\Phi (x)}\Big)^{n}, & n \text{ even}.
\end{array}
\right.
\end{equation}
By similar calculation, the set $\wt{\mathcal S}_n=\{\wt{\mathcal Y}_0(x_0,x),\ldots ,\wt{\mathcal Y}_n(x_0,x)\}$ is a fundamental set of solutions for the ordinary differential equations $D^{(n+1)}y(x)=0$ with $n$ even and $\wt D^{(n+1)}y(x)=0$ with $n$ odd. We find
\begin{equation}
\label{Wtn}
\wt{\mathcal W}_n(x)=\left\{\begin{array}{ll}
\alpha_n\Big(\displaystyle\frac{1}{\sqrt{\Phi (x)}}\Big)^{n+1}, & n \text{ odd},\\*[3ex]
\alpha_n \Big(\displaystyle\frac{1}{\sqrt{\Phi (x)}}\Big)^{n}, & n \text{ even}.
\end{array}
\right.
\end{equation}

The function $\wt X^{(n)}/\Phi$ (resp. $X^{(n)}\Phi$) is the function of Cauchy \cite{goursat} used in obtaining a particular solution of the nonhomogneous equation $D^{(n+1)}y(x)=h(x)$ (resp. $\wt D^{(n+1)}y(x)=h(x)$) from their solution of the corresponding homogeneous equation. The particular solutions $y_p(x)$ is given by
$$
y_p(x)=\displaystyle\frac{1}{n!}\int_{x_0}^x \wt X^{(n)}(\xi,x)\frac{h(\xi)}{\Phi (\xi)}d\xi \qquad \text{for equation} \qquad D^{(n+1)}y(x)=h(x)
$$
and
$$
y_p(x)=\displaystyle\frac{1}{n!}\int_{x_0}^x X^{(n)}(\xi,x)\Phi (\xi)h(\xi)d\xi \qquad \text{for equation}\qquad \wt D^{(n+1)}y(x)=h(x).
$$
Indeed, we have
\begin{align*}
D^{(n+1)}y_p(x)&=D\wt D^{(n)}y_p(x)=D\Big[\displaystyle\frac{1}{n!}\int_{x_0}^x \Big(\wt D^{(n)}\wt X^{(n)}(\xi,x)\Big)\frac{h(\xi)}{\Phi (\xi)}d\xi\Big]\\*[2ex]
&=D\int_{x_0}^x \frac{h(\xi)}{\Phi (\xi)}d\xi=h(x).
\end{align*}
Similar calculations are made to show a particular solution of $\wt D^{(n+1)}y(x)=h(x)$.

\begin{theorem}[Supersymmetric $\Phi $-Taylor series]
\label{Taylor}
Let the functions $f(x),\mathcal Y_0(x),\mathcal Y_1(x),\ldots,\mathcal Y_n(x)$ be real-valued functions of class $C^{n+1}$ in the interval $[a,b]$ and $a\leq x_0\leq b$. Then
\begin{equation}
\label{taylor}
f(x)=\sum_{k=0}^n \frac{\mathcal D_kf(x_0)}{k!}\mathcal Y_k(x_0,x)+R_n(x),
\end{equation}
where 
$$
\mathcal D_kf(x)=\left\{
\begin{array}{ll}
\wt D^{(k)}f(x),& k\text{ odd},\\*[2ex]
D^{(k)}f(x), & k\text{ even},
\end{array}
\right. \quad \text{and}\quad R_n(x)= \frac{1}{n!}\int_{x_0}^x \big[\Phi (\xi)\big]^{(-1)^n}\mathcal Y_{n}(\xi,x)\mathcal D_{n+1}f(\xi)d\xi.
$$
\end{theorem}

\begin{proof}
From Proposition \ref{linind} the Wronskian $\mathcal W_n(x)>0$ in the interval $[a,b]$. Now considering the result obtained by Widder  \cite[p.138]{widder1}, we have 
\begin{equation}
f(x)=\sum_{k=0}^n L_kf(x_0)\,g_k(x_0,x)+R_n(x),
\end{equation}
where the functions $g_k(x_0,x)$ are defined as $g_0(x_0,x)=1$ and
\begin{equation}
g_k(x_0,x)=\left(\frac{1}{\mathcal W_k(x_0)}\right)\left|
\begin{array}{cccc}
1 & \mathcal Y_1(x_0,x_0) & \cdots & \mathcal Y_k(x_0,x_0)\\
0 & \mathcal Y_1^{(1)}(x_0,x_0) & \cdots & \mathcal Y_k^{(1)}(x_0,x_0)\\
\vdots & \vdots & \ddots & \vdots \\
0 & \mathcal Y_1^{(k-1)}(x_0,x_0) & \cdots & \mathcal Y_k^{(k-1)}(x_0,x_0)\\
1 & \mathcal Y_1(x_0,x) & \cdots & \mathcal Y_k(x_0,x)
\end{array}
\right|,
\end{equation}
while the operators $L_k$ are defined as $L_0=\text{id}$ and
\begin{equation}
L_kf(x)=\frac{\mathcal W[\mathcal Y_0(x_0,x),\mathcal Y_1(x_0,x),\ldots,\mathcal Y_{k-1}(x_0,x),f(x)]}{\mathcal W_{k-1}(x)}.
\end{equation}
However, these functions $g_k$ and operators $L_k$ defined in terms of a determinant
 can be expressed in a more convenient way in terms of functions $\phi_k(x)$ given by
$$
\phi_0(x)=1,\qquad \phi_1(x)=\Phi (x),\quad \phi_i(x)=\frac{\mathcal W_i(x)\mathcal W_{i-2}(x)}{\big[\mathcal W_{i-1}(x)\big]^2},\quad 2\leq i\leq n-1.
$$
Indeed, from \cite{widder1} we have
$$
\begin{array}{rcl}
g_k(x_0,x)&=& \displaystyle\frac{\phi_0(x)}{\phi_0(x_0)\cdots \phi_k(x_0)}\displaystyle\int_{x_0}^x \phi_1(\xi_1) \displaystyle\int_{x_0}^{\xi_1}\\*[3ex]
&& \cdots \displaystyle\int_{x_0}^{\xi_{k-2}} \phi_{k-1}(\xi_{k-1}) \displaystyle\int_{x_0}^{\xi_{k-1}}\phi_k(\xi_k)d\xi_k d\xi_{k-1}\cdots d\xi_1
\end{array}
$$
and
$$
L_kf(x)=\phi_0(x)\phi_1(x)\cdots \phi_{k-1}(x)\frac{d}{dx}\frac{1}{\phi_{k-1}(x)} \frac{d}{dx}\frac{1}{\phi_{k-2}(x)}\frac{d}{dx}\cdots \frac{d}{dx}\frac{1}{\phi_{1}(x)} \frac{d}{dx}\frac{f(x)}{\phi_{0}(x)}.
$$
In our case the functions $\phi_i(x)$ are calculated from equation \eqref{Wn}. We obtain $\phi_0(x)=1$ and $\phi_k(x)=k\left(\displaystyle\frac{1}{\Phi (x)}\right)^{(-1)^k}$. It is now easy to identify the functions $g_k(x_0,x)$ and the operators $L_k$ as
$$
g_k(x_0,x)=\left\{
\begin{array}{ll}
\displaystyle\frac{1}{k!}\frac{1}{\Phi (x_0)}\wt X^{(k)}(x_0,x),& k\text{ odd},\\*[2ex]
\displaystyle\frac{1}{k!}X^{(k)}(x_0,x),& k\text{ even},
\end{array}
\right.
$$
and
\begin{equation}
\label{Lk}
L_kf(x)= \left\{
\begin{array}{ll}
\Phi (x)\wt D^{(k)}f(x),& k\text{ odd},\\*[2ex]
D^{(k)}f(x),& k\text{ even},
\end{array}
\right.
\end{equation}
for $k\geq 0$.

Now considering the rest, we have that $R_n(x)$ is of the form
$$
R_n(x)=\int_{x_0}^x g_n(\xi,x)L_{n+1}f(\xi)d\xi,
$$
which is equivalent to write
$$
R_n(x)=\frac{1}{n!}\int_{x_0}^x \big[\Phi (\xi)\big]^{(-1)^n}\mathcal Y_{n}(\xi,x)\mathcal D_{n+1}f(\xi)d\xi.
$$
\end{proof}

The proof of Theorem \ref{Taylor} gives us the possibility to obtain formulas of the coefficients $a_k(x),\wt a_k(x)$ for the developpement of the generalized derivatives $D^{(n)}y(x),\wt D^{(n)}y(x)$ in terms of the usual derivatives, i.e.
$$
D^{(n)}f(x)=\sum_{k=1}^{n} a_k(x)\,f^{(k)}(x) \qquad \text{ and } \qquad \wt D^{(n)}f(x)=\sum_{k=1}^{n} \wt a_k(x)\,f^{(k)}(x).
$$
Indeed, from \eqref{Lk} and its associate $\Phi $-conjugaison, we have
$$
D^{(n)}f(x)=\left\{
\begin{array}{ll}
\Phi (x)\wt L_nf(x), & n \text{ odd},\\*[2ex]
L_nf(x), & n\text{ even},
\end{array}
\right. \quad \text{and} \quad \wt D^{(n)}f(x)=\left\{
\begin{array}{ll}
\displaystyle\frac{1}{\Phi (x)} L_nf(x), & n \text{ odd},\\*[2ex]
\wt L_nf(x), & n\text{ even}.
\end{array}
\right.
$$

Considering first the case $n>0$ and even we find
\begin{align*}
D^{(n)}f(x)&=\frac{1}{\mathcal W_{n-1}(x)}
\left|
\begin{array}{ccccc}
\mathcal Y_0 & \mathcal Y_1 & \hdots & \mathcal Y_{n-1} & f \\
\mathcal Y_0' & \mathcal Y_1' & \hdots & \mathcal Y_{n-1}' & f' \\
\vdots & \vdots & \ddots & \vdots & \vdots\\ 
\mathcal Y_0^{(n)} & \mathcal Y_1^{(n)} & \hdots & \mathcal Y_{n-1}^{(n)} & f^{(n)} 
\end{array}
\right|\\*[2ex]
&= \frac{1}{\alpha_{n-1}\big[\Phi (x)\big]^{\frac{n}{2}}} \left|
\begin{array}{ccccc}
 \mathcal Y_1' & \mathcal Y_2' & \hdots & \mathcal Y_{n-1}' & f' \\
\mathcal Y_1'' & \mathcal Y_2'' & \hdots & \mathcal Y_{n-1}'' & f'' \\
\vdots & \vdots &\ddots & \vdots & \vdots\\ 
\mathcal Y_1^{(n)} & \mathcal Y_2^{(n)} & \hdots & \mathcal Y_{n-1}^{(n)} & f^{(n)} 
\end{array}
\right| \\*[2ex]
&= \frac{1}{\alpha_{n-1}\big[\Phi (x)\big]^{\frac{n}{2}}}\sum_{k=1}^n \varepsilon_{i_1,\ldots, i_{n-1},k} \left|
\begin{array}{ccccc}
 \mathcal Y_1^{(i_1)} & \mathcal Y_2^{(i_1)} & \hdots & \mathcal Y_{n-1}^{(i_1)} \\
\mathcal Y_1 ^{(i_2)} & \mathcal Y_2^{(i_2)} & \hdots & \mathcal Y_{n-1}^{(i_2)} \\
\vdots & \vdots & \ddots & \vdots \\ 
\mathcal Y_1^{(i_{n-1})} & \mathcal Y_2^{(i_{n-1})} & \hdots & \mathcal Y_{n-1}^{(i_{n-1})} 
\end{array}
\right|f^{(k)}(x)
\end{align*}
where $1\leq i_1<i_2<\cdots <i_{n-1}\leq n$
and $\varepsilon_{i_1i_2\cdots i_{n-1}k}$ represents the $n$-dimensional Levi-Cevita symbol. By a similar calculation for $n$ odd, we finally obtain the following results for the coefficients $a_k(x)$ (coefficients $\wt a_k(x)$ are obtained by $\Phi $-conjugaison).

\begin{proposition}
The $n$-order $\Phi $-generalized derivatives $D^{(n)}f(x)$ and $\wt D^{(n)}f(x)$ can be expanded in terms of the usual derivatives $D^{(n)}f(x)=\sum_{k=1}^{n} a_k(x)\,f^{(k)}(x)$ and $\wt D^{(n)}f(x)=\sum_{k=1}^{n} \wt a_k(x)\,f^{(k)}(x)$, where
$$
a_k(x)=\left\{
\begin{array}{ll}
\varepsilon_{i_1,\ldots, i_{n-1},k}\displaystyle\frac{\big[\Phi (x)\big]^{\frac{n+1}{2}}}{\alpha_{n-1}} \left|
\begin{array}{ccccc}
 \wt{\mathcal Y}_1^{(i_1)} & \wt{\mathcal Y}_2^{(i_1)} & \hdots & \wt{\mathcal Y}_{n-1}^{(i_1)} \\
\wt{\mathcal Y}_1 ^{(i_2)} & \wt{\mathcal Y}_2^{(i_2)} & \hdots & \wt{\mathcal Y}_{n-1}^{(i_2)} \\
\vdots & \vdots & \ddots & \vdots \\ 
\wt{\mathcal Y}_1^{(i_{n-1})} & \wt{\mathcal Y}_2^{(i_{n-1})} & \hdots & \wt{\mathcal Y}_{n-1}^{(i_{n-1})} 
\end{array}
\right|, & n\text{ odd}, \\*[7ex]
\varepsilon_{i_1,\ldots, i_{n-1},k}\displaystyle\frac{1}{\alpha_{n-1} \big[\Phi (x)\big]^{\frac{n}{2}}} \left|
\begin{array}{ccccc}
 \mathcal Y_1^{(i_1)} & \mathcal Y_2^{(i_1)} & \hdots & \mathcal Y_{n-1}^{(i_1)} \\
\mathcal Y_1 ^{(i_2)} & \mathcal Y_2^{(i_2)} & \hdots & \mathcal Y_{n-1}^{(i_2)} \\
\vdots & \vdots & \ddots & \vdots \\ 
\mathcal Y_1^{(i_{n-1})} & \mathcal Y_2^{(i_{n-1})} & \hdots & \mathcal Y_{n-1}^{(i_{n-1})} 
\end{array}
\right|, & n\text{ even},
\end{array}
\right.
$$
$$
\wt a_k(x)=\left\{
\begin{array}{ll}
\varepsilon_{i_1,\ldots, i_{n-1},k}\displaystyle\frac{1}{\alpha_{n-1}\big[\Phi (x)\big]^{\frac{n+1}{2}}} \left|
\begin{array}{ccccc}
 \mathcal Y_1^{(i_1)} & \mathcal Y_2^{(i_1)} & \hdots & \mathcal Y_{n-1}^{(i_1)} \\
\mathcal Y_1 ^{(i_2)} & \mathcal Y_2^{(i_2)} & \hdots & \mathcal Y_{n-1}^{(i_2)} \\
\vdots & \vdots & \ddots & \vdots \\ 
\mathcal Y_1^{(i_{n-1})} & \mathcal Y_2^{(i_{n-1})} & \hdots & \mathcal Y_{n-1}^{(i_{n-1})} 
\end{array}
\right|, & n\text{ odd}, \\*[7ex]
\varepsilon_{i_1,\ldots, i_{n-1},k}\displaystyle\frac{\big[\Phi (x)\big]^{\frac{n}{2}}}{\alpha_{n-1}} \left|
\begin{array}{ccccc}
 \wt{\mathcal Y}_1^{(i_1)} &  \wt{\mathcal Y}_2^{(i_1)} & \hdots &  \wt{\mathcal Y}_{n-1}^{(i_1)} \\
 \wt{\mathcal Y}_1 ^{(i_2)} &  \wt{\mathcal Y}_2^{(i_2)} & \hdots &  \wt{\mathcal Y}_{n-1}^{(i_2)} \\
\vdots & \vdots & \ddots & \vdots \\ 
 \wt{\mathcal Y}_1^{(i_{n-1})} &  \wt{\mathcal Y}_2^{(i_{n-1})} & \hdots &  \wt{\mathcal Y}_{n-1}^{(i_{n-1})} 
\end{array}
\right|, & n\text{ even},
\end{array}
\right.
$$

where 
$$
1\leq i_1<i_2<\cdots <i_{n-1}\leq n \qquad \text{and}\qquad 1\leq k\leq n.
$$
\end{proposition}

\section{Volterra composition of the $\Phi $-power functions}
Let us first introduce a composition introduced by Volterra \cite{volterra2005theory}. For $f(x,y)$ and $g(x,y)$ two functions such that 
\begin{equation}
\label{volterraint}
\int_x^y f(x,\xi)g(\xi,y)\,d\xi.
\end{equation}
exists for $a\leq x\leq y\leq b$. The integral \eqref{volterraint} is called the {\em composition of the first type} and is denoted by $\big(f\star g\big)(x,y)$.

As shown in \cite{volterra2005theory}, the composition of the first type is associative and distributive.  Moreover, if $f\star g=g\star f$ then the two functions $f$ and $g$ are said to be {\em permutable}. The following notation will be used when the function $f$ is composed with itself:
\begin{equation*}
f^{\angl{n}}= \big(f^{\angl{n-1}}\star f\big),\qquad n\geq 1.
\end{equation*}
Here $f^{\angl{0}}(x,y)=\delta(y-x)$ and $f^{\angl{1}}(x,y)=f(x,y)$ by definition, where $\delta(y-x)$ is the  Dirac delta function. Hence, it is easy to show that functions $f^{\angl{n}}$ and $f^{\angl{m}}$ are permutable for some positive integers $m,n$.

Let us now introduce the following functions
\begin{equation*}
\mathbf{1}(x,y) \equiv 1 \qquad \text{ and }\qquad \sigma(x,y) \equiv \frac{\Phi (y)}{\Phi (x)}.
\end{equation*}

By direct calculations, we find that 
\begin{align*}
\big(\mathbf{1}\star \sigma\big)(x_0,x)&=\Phi (x) \X{1}{x_0,x}, & \big(\mathbf{1}\star \sigma^{-1}\big)(x_0,x)&=\frac{1}{\Phi (x)}\Xt{1}{x_0,x} \\*[2ex]
\big(\mathbf{1}\star \sigma\star \mathbf 1\big)(x_0,x)&=\frac{\X{2}{x_0,x}}{2!}, & \big(\mathbf{1}\star \sigma^{-1}\star \mathbf 1\big)(x_0,x)&=\frac{\Xt{2}{x_0,x}}{2!} \\*[2ex]
\big(\mathbf{1}\star \sigma\big)^{\angl{2}}(x_0,x)&=\Phi (x)\frac{\X{3}{x_0,x}}{3!}, & \big(\mathbf{1}\star \sigma^{-1}\big)^{\angl{2}}(x_0,x)&=\frac{1}{\Phi (x)}\frac{\Xt{3}{x_0,x}}{3!}, \ldots
\end{align*}
Taking into account that $\big(\mathbf{1}\star \sigma\big)=\Phi X^{(1)}$ and  $\big(\mathbf{1}\star \sigma^{-1}\big)=\displaystyle \frac{1}{\Phi }\wt X^{(1)}$, we have  for $n\geq 1$ 
\begin{align}
\label{psiX1}
\Big(\Phi X^{(1)}\Big)^{\angl{n}}=\Phi \frac{X^{(2n-1)}}{(2n-1)!}, \qquad \Big(\frac{1}{\Phi } \wt X^{(1)}\Big)^{\angl{n}}=\frac{1}{\Phi }\frac{\wt X^{(2n-1)}}{(2n-1)!}
\end{align}
and
\begin{align}
\label{psiX2}
\Big(\Phi X^{(1)}\Big)^{\angl{n}}\star \mathbf 1=\frac{X^{(2n)}}{(2n)!}, \qquad \Big(\frac{1}{\Phi } \wt X^{(1)}\Big)^{\angl{n}}\star \mathbf 1=\frac{\wt X^{(2n)}}{(2n)!}.
\end{align}

\begin{proposition}
The following identities are valid for the Volterra composition of the first type for the $\Phi $-power functions:
\label{multiplication}
\begin{align*}
\Phi \Xn{2n-1}\star \Xn{2m}=  A_{n,m} \Xn{2n+2m},\quad \frac{1}{\Phi }\Xtn{2n-1}\star \Xtn{2m} = A_{n,m} \Xtn{2n+2m},\\
\Phi \Xn{2n-1}\star \Xn{2m-1}= B_{n,m} \Xn{2n+2m-1},\quad \frac{1}{\Phi }\Xtn{2n-1}\star \Xtn{2m-1} = B_{n,m} \Xtn{2n+2m-1}, 
\end{align*}
where
\begin{equation*}
A_{n,m} := \frac{(2n-1)!(2m)!}{(2n+2m)!}, \qquad  B_{n,m} := \frac{(2n-1)!(2m-1)!}{(2n+2m-1)!},\qquad n,m\geq 1.
\end{equation*}
\end{proposition}
\begin{proof}
Let us calculate the first identity. We have
\begin{align*}
\Phi  X^{(2n-1)}\star X^{(2m)}&=(2n-1)!\big(\Phi  X^{(1)}\big)^{\angl{n}}\star (2m)!\big[\big(\Phi  X^{(1)}\big)^{\angl{m}}\star \mathbf 1\big] \\*[2ex]
&= (2n-1)!(2m)!\big[\big(\Phi  X^{(1)}\big)^{\angl{n+m}}\star \mathbf 1\big] \\*[2ex]
&= (2n-1)!(2m)! \frac{X^{(2n+2m)}}{(2n+2m)!}.
\end{align*}
Other identities are calculated in a similar way.
\end{proof}

In \cite{volterra2005theory} Volterra has shown that given any analytic function $\sum_{n=0}^\infty c_n z^n$ which converges within a certain circle and given any bounded function $f(x,y)$, then $\sum_{n=0}^\infty c_n f^{\angl{n}}$
is convergent for all values of $f(x,y)$. In other words, there is an isomorphism between algebraic formulas which involve only addition and multiplication and those obtained from the algebraic formulas by replacing the powers of the variable with  the powers by composition of the corresponding function $f$. It follows that if $F(z)$ and $G(z)$ are two analytic functions and $F(z)\,G(z)=K(z)$ then $F(f^{\angl{1}}) \star G(f^{\angl{1}})=K(f^{\angl{1}})$.

\begin{theorem}
\label{repsolvol}
Assume that on a finite real interval $[a,b]$, equation $-\psi^{\prime\prime}_0(x)+V(x)\psi_0(x) =0 $ possesses a particular solution $\psi_0(x)$ such that $\psi_0^2(x)$ and $1/\psi_0^2(x)$ are continuous on $[a, b]$. Then the general solution of the Schr\"odinger equation $-\psi^{\prime\prime}(x)+V(x)\psi(x)=\lambda\psi(x)$ on $(a, b)$ has the form 
\begin{equation}
\label{gensolsch}
\psi(x)=c_1\psi_0 \int_{x_0}^x \frac{1}{1-\lambda \wt \rho^{\angl{1}}}d\xi
+\frac{c_2}{\psi_0}\left(\frac{\rho^{\angl{1}}}{1-\lambda \rho^{\angl{1}}}\right),
\end{equation}
where 
$$
\rho(x_0,x)=\psi_0^2(x)\int_{x_0}^x \frac{d\xi}{\psi_0^2(\xi)},\qquad   \wt \rho(x_0,x)=\frac{1}{\psi_0^2(x)}\int_{x_0}^x \psi_0^2(\xi) d\xi,
$$
$c_1,c_2$ are two arbitrary complex constants and $x_0$ is an arbitrary fixed point in $[a,b]$. The geometric series in \eqref{gensolsch} converge for every $x$ on the interval $[a,b]$.
\end{theorem}
\begin{proof}
We set $\Phi=\psi_0^2$. Now, from Theorem \ref{kravch}  and   equations \eqref{psiX1},  \eqref{psiX2}, the general solution of $-\psi^{\prime\prime}(x)+V(x)\psi=\lambda\psi(x)$  is given by
\begin{align}
\label{part1}
\psi&=c_1\psi_0  \sum_{k=0}^\infty \lambda^k \frac{\wt X^{(2k)}}{(2k)!}+c_2\psi_0  \sum_{k=0}^\infty \lambda^k \frac{X^{(2k+1)}}{(2k+1)!} \nonumber \\*[2ex]
&=c_1\psi_0 \left\{1+\sum_{k=1}^\infty \lambda^k \left[\left(\frac{1}{\psi_0^2 }\wt X^{(1)}\right)^{\angl{k}}\star \mathbf 1\right]\right\}+c_2\psi_0 \sum_{k=0}^\infty \lambda^k \frac{1}{\psi_0^2 }\left(\psi_0^2  X^{(1)}\right)^{\angl{k+1}} \nonumber \\*[2ex]
&=c_1\psi_0 \left\{1+\sum_{k=1}^\infty \lambda^k \left[\wt \rho^{\angl{k}}\star \mathbf 1\right]\right\}+\frac{c_2}{\psi_0} \sum_{k=0}^\infty \lambda^k \rho^{\angl{k+1}},
\end{align}
where the last equality is valid since
$$
\frac{1}{\psi_0^2}\wt X^{(1)}=\frac{1}{\psi_0^2}\int_{x_0}^x \psi_0^2(\xi)d\xi=\wt \rho,\qquad \psi_0^2 X^{(1)}=\psi_0^2\int_{x_0}^x \frac{d\xi}{\psi_0^2(\xi)}=\rho.
$$
Now, since $\wt \rho^{\angl{0}}\star \mathbf 1=\delta(x-x_0)\star \mathbf 1=1$ and the Volterra composition is associative, equation \eqref{part1} becomes
\begin{align}
\label{eqschvolt}
\psi&=c_1\psi_0\left[ \left(\sum_{k=0}^\infty \lambda^k \wt \rho^{\angl{k}}\right)\star \mathbf 1\right]+\frac{c_2}{\psi_0} \sum_{k=0}^\infty \lambda^k \rho^{\angl{k+1}}.
\end{align}
Considering the geometric series
$$
\Omega(\lambda,z)=\sum_{k=0}^\infty \lambda^k z^k=\frac{1}{1-\lambda z},
$$
where for the usual multiplication the last equality is valid for $|\lambda z|<1$ . However, for the Volterra composition $\star$ equation \eqref{eqschvolt} becomes
\begin{align}
u&=c_1\psi_0\left( \Omega(\lambda,\wt \rho^{\angl{1}})\star \mathbf 1\right)+\frac{c_2}{\psi_0}\left(\rho^{\angl{1}}\star\Omega(\lambda,\rho^{\angl{1}})\right)\nonumber \\*[2ex] 
&=c_1\psi_0\left(\frac{1}{1-\lambda \wt \rho^{\angl{1}}}\star \mathbf 1\right)+\frac{c_2}{\psi_0}\left(\frac{\rho^{\angl{1}}}{1-\lambda \rho^{\angl{1}}}\right) 
\label{quasifin}
\end{align}
and the geometric series $\Omega(\lambda,\wt \rho^{\angl{1}})$,  $\Omega(\lambda,\rho^{\angl{1}})$ converge for all values of $x_0,x$ and $\lambda$. Furthermore, by definition of the Volterra composition the first term of equation \eqref{quasifin} is given by
$$
c_1\psi_0\left(\frac{1}{1-\lambda \wt \rho^{\angl{1}}}\star \mathbf 1\right)=c_1\psi_0 \int_{x_0}^x \frac{1}{1-\lambda \wt \rho^{\angl{1}}(x_0,\xi)}d\xi
$$
which complete the proof.
\end{proof}

\begin{remark}
The solution representation $\wt \psi(x)$ of the supersymmetric partner Schr\"odinger equation can be obtained by applying operator $\mathcal R_{\Phi}$ in Theorem \ref{repsolvol}, i.e.
$$
\wt \psi(x)=\frac{\wt c_1}{\psi_0} \int_{x_0}^x \frac{1}{1-\lambda \rho^{\angl{1}}}d\xi
+\wt c_2\psi_0\left(\frac{\wt \rho^{\angl{1}}}{1-\lambda \wt \rho^{\angl{1}}}\right).
$$

\end{remark}

\subsubsection*{Acknowledgement}
M. O. acknowledges scholarships from NSERC and FRQNT, where a part of this work was done, and one from the Institut des Sciences Math\'ematiques (ISM). S. T. thanks the CRM, Universit\'e de Montr\'eal, for its hospitality where part of this work was done during his sabbatical year.

\newpage

\bibliography{bib_art1_math}{}

\begin{thebibliography}{10}

\bibitem{begehr}
Heinrich G.~W. Begehr and Robert~P. Gilbert.
\newblock {\em Transformations, transmutations and kernel functions}, volume
  1,2.
\newblock Longman Scientific \& Technical, 1992.

\bibitem{bilodeautremblay}
Alex Bilodeau and Sébastien Tremblay.
\newblock On two-dimensional supersymmetric quantum mechanics, pseudoanalytic
  functions and transmutation operators.
\newblock {\em Journal of Physics A: Mathematical and Theoretical},
  46(42):425302, 2013.

\bibitem{blancarte}
Herminio Blancarte, Hugo~M. Campos, and Kira~V. Khmelnytskaya.
\newblock Spectral parameter power series method for discontinuous
  coefficients.
\newblock {\em Mathematical Methods in the Applied Sciences},
  38(10):2000--2011, 2015.

\bibitem{carroll}
Robert~W. Carroll.
\newblock {\em Transmutation theory and applications}, volume~11.
\newblock North-Holland, 1985.

\bibitem{castillo2015}
Ra{\'{u}}l Castillo-P{\'{e}}rez, Vladislav~V. Kravchenko, and Sergii~M. Torba.
\newblock Analysis of graded-index optical fibers by the spectral parameter
  power series method.
\newblock {\em Journal of Optics}, 17(2):025607, 2015.

\bibitem{castillo2013}
Raúl Castillo-Pérez, Vladislav~V. Kravchenko, and Sergii~M. Torba.
\newblock Spectral parameter power series for perturbed {B}essel equations.
\newblock {\em Applied Mathematics and Computation}, 220:676 -- 694, 2013.

\bibitem{chen1977}
Kuo-Tsai Chen.
\newblock Iterated path integrals.
\newblock {\em Bulletin of the American Mathematical Society}, 83(5):831--879,
  1977.

\bibitem{colton}
David~L. Colton.
\newblock {\em Solution of boundary value problems by the method of integral
  operators}, volume~11.
\newblock Pitman, 1976.

\bibitem{davis}
Philip~J. Davis.
\newblock {\em Interpolation and Approximation}.
\newblock Dover Publications, 1975.

\bibitem{erbe}
Lynn Erbe, Raziye Mert, and Allan Peterson.
\newblock Spectral parameter power series for {S}turm–{L}iouville equations
  on time scales.
\newblock {\em Applied Mathematics and Computation}, 218(14):7671 -- 7678,
  2012.

\bibitem{goursat}
E.~Goursat.
\newblock {\em Cours d'analyse math{\'e}matique, Vol. 2: Th{\'e}orie des
  fonctions analytiques; {\'E}quations diff{\'e}rentielles; {\'E}quations aux
  d{\'e}riv{\'e}es partielles du premier ordre (Classic Reprint)}.
\newblock Fb\&c Limited, 2018.

\bibitem{han2016}
Zheng Han, Yaozhong Hu, and Chihoon Lee.
\newblock Optimal pricing barriers in a regulated market using reflected
  diffusion processes.
\newblock {\em Quantitative Finance}, 16(4):639--647, apr 2016.

\bibitem{khmelnytskaya2012}
Kira~V. Khmelnytskaya, Vladislav~V. Kravchenko, and Jes\'us~A.
  Baldenebro-Obeso.
\newblock Spectral parameter power series for fourth-order sturm–liouville
  problems.
\newblock {\em Applied Mathematics and Computation}, 219(8):3610 -- 3624, 2012.

\bibitem{khmelnytskaya2014}
Kira~V. Khmelnytskaya, Vladislav~V. Kravchenko, and Haret~C. Rosu.
\newblock Eigenvalue problems, spectral parameter power series, and modern
  applications.
\newblock {\em Mathematical Methods in the Applied Sciences}, 2014.

\bibitem{khmelnytskaya2010}
Kira~V. Khmelnytskaya and Haret~C. Rosu.
\newblock Spectral parameter power series representation for hill’s
  discriminant.
\newblock {\em Annals of Physics}, 325(11):2512 -- 2521, 2010.

\bibitem{khmelnytskaya2013}
Kira~V. Khmelnytskaya and Ibrahim Serroukh.
\newblock The heat transfer problem for inhomogeneous materials in
  photoacoustic applications and spectral parameter power series.
\newblock {\em Mathematical Methods in the Applied Sciences},
  36(14):1878--1891, 2013.

\bibitem{kravchenko2008}
Vladislav~V. Kravchenko.
\newblock A representation for solutions of the {S}turm–{L}iouville equation.
\newblock {\em Complex Variables and Elliptic Equations}, 53(8):775--789, 2008.

\bibitem{kravchenko2011completeness}
Vladislav~V. Kravchenko.
\newblock On the completeness of systems of recursive integrals.
\newblock {\em Communications in Mathematical Analysis}, (3):172--176, 2011.

\bibitem{kravchenko2018}
Vladislav~V. Kravchenko.
\newblock Construction of a transmutation for the one-dimensional schrödinger
  operator and a representation for solutions.
\newblock {\em Applied Mathematics and Computation}, 328:75 -- 81, 2018.

\bibitem{kravchenko2016}
Vladislav~V. Kravchenko, Samy Morelos, and Sergii~M. Torba.
\newblock Liouville transformation, analytic approximation of transmutation
  operators and solution of spectral problems.
\newblock {\em Applied Mathematics and Computation}, 273:321 -- 336, 2016.

\bibitem{kravtrembmorelos}
Vladislav~V. Kravchenko, Samy Morelos, and S\'ebastien Tremblay.
\newblock Complete systems of recursive integrals and taylor series for
  solutions of {S}turm-{L}iouville equations.
\newblock {\em Mathematical Methods in the Applied Sciences}, 35(6):704--715,
  2011.

\bibitem{kravporter2010}
Vladislav~V Kravchenko and R~Michael Porter.
\newblock Spectral parameter power series for {S}turm--{L}iouville problems.
\newblock {\em Mathematical Methods in the Applied Sciences}, 33(4):459--468,
  2010.

\bibitem{kravchenko2013}
Vladislav~V. Kravchenko and Sergii~M. Torba.
\newblock {\em Transmutations and spectral parameter power series in eigenvalue
  problems}, pages 209--238.
\newblock Springer Basel, Basel, 2013.

\bibitem{kravchenko2015}
Vladislav~V. Kravchenko and Sergii~M. Torba.
\newblock Analytic approximation of transmutation operators and applications to
  highly accurate solution of spectral problems.
\newblock {\em Journal of Computational and Applied Mathematics}, 275:1 -- 26,
  2015.

\bibitem{castillo2018}
Vladislav~V. Kravchenko, Sergii~M. Torba, and Raúl Castillo-Pérez.
\newblock A neumann series of {B}essel functions representation for solutions
  of perturbed {B}essel equations.
\newblock {\em Applicable Analysis}, 97(5):677--704, 2018.

\bibitem{levitan}
Boris~M. Levitan.
\newblock {\em Inverse Sturm-Liouville problems}.
\newblock VNU Science Press, 1987.

\bibitem{marchenko}
Vladimir Marchenko.
\newblock {\em Sturm-Liouville Operators and Applications}.
\newblock Birkhäuser Basel, 1986.

\bibitem{porter2016}
Michael~R. Porter.
\newblock On {S}turm-{L}iouville equations with several spectral parameters.
\newblock {\em Bolet{\'i}n de la Sociedad Matem{\'a}tica Mexicana},
  22(1):141--163, 2016.

\bibitem{rabinovich2007}
Vladimir Rabinovich and Francisco Urbano~Altamirano.
\newblock {\em Application of the SPPS method to the one-dimensional quantum
  scattering}, volume~17.
\newblock Communications in Mathematical Analysis, 2007.

\bibitem{rabinovich2015}
Vladimir~S. Rabinovich and Josu{\'{e}} Hern{\'{a}}ndez-Ju{\'{a}}rez.
\newblock Method of the spectral parameter power series in problems of
  underwater acoustics of the stratified ocean.
\newblock {\em Mathematical Methods in the Applied Sciences},
  38(10):1990--1999, 2015.

\bibitem{volterra2005theory}
Vito Volterra.
\newblock {\em Theory of functionals and of integral and integro-differential
  equations}.
\newblock Courier Corporation, 2005.

\bibitem{widder1}
D.~V. Widder.
\newblock A generalization of {T}aylor's series.
\newblock {\em Transactions of the American Mathematical Society},
  30(1):126--154, 1928.

\bibitem{witten}
Edward Witten.
\newblock Dynamical breaking of supersymmetry.
\newblock {\em Nuclear Physics B}, 188(3):513 -- 554, 1981.

\end{thebibliography}
\bibliographystyle{plain}

\end{document}